\documentclass[a4paper,12pt]{article}
\usepackage{amsmath, amssymb, amsthm}
\usepackage{bbm}
\usepackage{natbib}
\usepackage{graphicx}
\usepackage{enumerate}
\usepackage[usenames]{color} 
\usepackage{exscale}
\usepackage{bm}
\usepackage[normalem]{ulem}
\usepackage{mathrsfs}

\definecolor{darkgreen}{rgb}{.2,.8,0}

\def\indf#1{\mathbbmss 1_{  #1  }} 

\def\real{\operatorname{Re}}
\def\imag{\operatorname{Im}}
\def\IdM{\text{Id}}
\def\bS{\mathbb{S}}
\def\bR{\mathbb{R}}
\def\bT{\mathbb{T}}
\def\bC{\mathbb{C}}

\def\wt#1{\widetilde #1}

\newtheorem{theorem}{Theorem}[section]
\newtheorem{lemma}[theorem]{Lemma}
\newtheorem{proposition}[theorem]{Proposition}
\newtheorem{corollary}[theorem]{Corollary}

\theoremstyle{remark}
\newtheorem{remark}[theorem]{Remark}

\theoremstyle{definition}
\newtheorem{definition}[theorem]{Definition}
\newtheorem{example}[theorem]{Example}

\renewcommand{\baselinestretch}{1}\normalsize

\begin{document}

\title{Multivariate transient price  impact and matrix-valued positive definite functions}

\author{Aur\'elien Alfonsi\footnote{Universit\'e Paris-Est, CERMICS,
   Projet MathRisk ENPC-INRIA-UMLV,
    Ecole des Ponts, 6-8 avenue Blaise Pascal,
   77455 Marne La Vall\'ee, France. {\tt alfonsi@cermics.enpc.fr
}}{\setcounter{footnote}{2}} \and Florian Kl\"ock\footnote{Department of Mathematics, University of Mannheim, 
A5, 6, 68131 Mannheim, Germany. 
}\and \setcounter{footnote}{6}Alexander
Schied\footnote{Department of Mathematics, University of Mannheim, 
A5, 6, 68131 Mannheim, Germany. 
{\tt schied@uni-mannheim.de}\hfill\break
The authors thank an anonymous referee for comments that helped to substantially improve a previous version of the manuscript.
A.A. is grateful for  the support of the \lq\lq Chaire Risques Financiers" of Fondation du Risque. F.K.~and A.S.~thank Martin Schlather and Marco Oesting for discussions and gratefully acknowledge financial support by Deutsche Forschungsgemeinschaft DFG through Research Grant SCHI 500/3-1.}}
\date{\small First version:  October 16, 2013\\
This version: September 9, 2015}

\maketitle
\vspace{-1cm}

\begin{abstract}
We consider a model for linear transient price impact for multiple assets that takes cross-asset impact into account. Our main goal is to single out properties that need to be imposed on the decay kernel so that the model admits well-behaved optimal trade execution strategies. We first show that the existence of such strategies is guaranteed by assuming that the decay kernel corresponds to a matrix-valued positive definite function. An example illustrates, however, that positive definiteness alone does not guarantee that optimal strategies are well-behaved. Building on previous results from the one-dimensional case, we investigate a class of  nonincreasing, nonnegative, and convex decay kernels with values in the symmetric $K\times K$ matrices. We show that these decay kernels are always positive definite and characterize when they are even strictly positive definite, a result that may be of independent interest. Optimal strategies for kernels from this class are  particularly well-behaved if one requires that the decay kernel is also commuting. We show how such decay kernels can be constructed by means of matrix functions and provide a number of examples. In particular, we completely solve the case of matrix exponential decay.\end{abstract}

\noindent{\bf Keywords:} Multivariate price impact, matrix-valued positive definite function, optimal trade execution, optimal portfolio liquidation, matrix function

\section{Introduction}

Price impact refers to the  feedback effect of trades on the quoted price of an asset and it is responsible for the creation of execution costs. It is an empirically established fact that price impact is predominantly transient; see, e.g., \citet{MoroEtAl}. When trading speed is sufficiently slow, the effects of transience  can be reduced 
to considering only a temporary and a permanent price impact component \citep{BertsimasLo,almgrenchriss2001}. For higher trading speeds, however, one needs a model that explicitly describes the decay of price impact between trades. First models of this type were proposed by \citet{bouchaud2004} and \citet{ObizhaevaWang}. These models were later extended into various directions by \citet{AFS1,AFS2}, \citet{Gatheral}, \citet{alfonsischiedslynko}, \citet{gatheralschiedslynko}, \citet{PredoiuShaikhetShreve}, \citet{FruthSchoenebornUrusov}, and \citet{Lokka},  to mention only a few. A more comprehensive list of references can be found in \citet{GatheralSchiedSurvey}. We also refer to \citet{GuoTutorial} for an introduction to the microscopic order book picture that is behind the mesoscopic models mentioned above.

All above-mentioned models for transient price impact deal only with one single risky asset. While multi-asset models for temporary and permanent price impact  \citep{schoeneborn2011} or for generic price impact functionals  \citep{schiedschoenebornteranchi,kratzschoeneborn} were considered earlier, we are not aware of any previous approaches to analyzing the specific effects of transient cross-asset price impact. Our goal in this paper is to propose and analyze a simple model for transient price impact between $K$ different risky assets. 
Following the one-dimensional ansatz of \citet{Gatheral}, the  time-$t$ impact on the price of the $i^{\text{th}}$ asset  that is generated by trading one unit of the $j^{\text{th}}$ asset at time $s<t$ will be described by the number $G_{ij}(t-s)$ for a certain function $G_{ij}:[0,\infty)\to\bR$.
The matrix-valued function $G(t)=(G_{ij}(t))_{i,j=1,\dots,K}$ will be called the \emph{decay kernel} of the multi-asset price impact model.
 
When setting up such a model in a concrete situation, the first question one encounters is how to choose the decay kernel. Already in the one-dimensional situation, $K=1$, the decay kernel $G$ needs to satisfy certain conditions so that the resulting price impact model has some minimal regularity properties such as the existence of optimal trade execution strategies, the absence of price manipulation in the sense of \citet{hubermanstanzl}, or the non-occurrence of oscillatory strategies. It was shown in \citet{alfonsischiedslynko} that these properties are satisfied when $G$ is nonnegative, nonincreasing, and convex. 
Here we will continue the corresponding analysis and   extend  it  to matrix-valued decay kernels $G$. Our first observation is that $G$ must correspond to a certain matrix-valued positive definite function. Such functions were previously characterized and analyzed, e.g., by \citet{cramer40,naimark43,falb69}. An example illustrates, however, that positive definiteness alone does not guarantee that optimal strategies are well-behaved. We therefore introduce a class of  nonincreasing, nonnegative, and convex decay kernels with values in the symmetric $K\times K$ matrices. We show that these decay kernels are always positive definite, and we characterize in Theorem \ref{convex strict pd} when they are even strictly positive definite. Optimal strategies for kernels from this class do not admit oscillations if one additionally requires that the decay kernel is commuting. Based on this result, we will  address in Section \ref{continuous-time section} the problem of optimizing simultaneously over  time grids and strategies  and state the solution in terms of a suitable continuous-time limit. We finally show how such decay kernels can be constructed by means of matrix functions  and provide a number of examples. In particular, we completely solve the case of matrix exponential decay. 

Our main general results are stated in Section \ref{section-2}. Transformation results for decay kernels and their optimal strategies 
along with several explicit examples are given in Section \ref{section-4}. Since the situation $K>1$ is considerably more complex than the one-dimensional case, we have summarized the main conclusions that can be drawn from our results in Section \ref{conclusion section}. 
These conclusions will focus on our initial question: From which class of functions should decay kernels for transient price impact be chosen?
Most proofs are given in Section \ref{sect-proofs}.

\section{Statement of general results}
\label{section-2}

In this section, we first introduce a linear market impact model with transient price impact for $K$ different risky assets.  We then discuss which properties a decay kernel should satisfy so that the corresponding market impact model has certain desirable features and properties. Two of these properties are the existence of optimal strategies and  the absence of price manipulation strategies in the sense of \citet{hubermanstanzl}, which we will both characterize by establishing a link to the theory of positive definite matrix-valued functions. Requiring positive definiteness, however, will typically not be sufficient to guarantee that optimal strategies are well-behaved. We will thus be led to a more detailed analysis of positive definite matrix-valued functions and the associated quadratic minimization problems, an analysis that might be of independent interest. 

\subsection{Preliminaries}

We introduce here a market impact model for an investor trading in $K$ different securities.  When the investor is not active, the \emph{unaffected price process} of these assets is given by a  right-continuous $K$-dimensional martingale $(S^0_t)_{t \in [0,T]}$ defined on a filtered probability space $(\Omega, \mathscr F, (\mathscr F_t)_{t\in[0,T]},\mathbb P)$. Now suppose that the  investor can trade at the times of a \emph{time grid} $\bT=\{t_1,\dots,t_N\}$, where $N\in\mathbb N$ and  $0 = t_1 < t_2 < \cdots < t_N $ (an extended setup with the possibility of trading in continuous time will be considered in Section \ref{continuous-time section}). The size of the order in the $i^{\text{th}}$ asset at time $t_k$ is described by a $\mathscr{F}_{t_k}$-measurable random variable $\xi^i_k$, where positive values denote buys and negative values denote sells. By $\xi_k=(\xi^1_k,\dots,\xi_k^K)^\top$ we denote the column vector of all orders placed at time $t_k$. Our main interest here will be in admissible strategies that $\mathbb P$-a.s.~liquidate a given initial portfolio $X_0\in\bR^K$. Such strategies are needed in practice when the initial portfolio $X_0$ is too big to be liquidated immediately; see, e.g., \citet{almgrenchriss2001}.


\begin{definition}\label{discrete-time def}Let $\bT=\{t_1,\dots,t_N\}$ be a time grid. An \emph{admissible strategy} for $\bT$ is a sequence $\bm\xi=(\xi_1,\ldots,\xi_N)$ of bounded\footnote{Boundedness is assumed here for simplicity and can easily be relaxed; for instance, it is enough to assume that  both $\xi_k$ and $S^0$ are square-integrable. Since the total number of shares of an asset is always finite, boundedness can be assumed without loss of generality from an economic point of view.} $K$-dimensional random variables such that each $\xi_k$ is $\mathscr{F}_{t_k}$-measurable;  $\bm\xi$ is called \emph{deterministic} if each  $\xi^i_k$ does not depend on $\omega\in\Omega$. The set of admissible liquidation strategies for a given initial portfolio $X_0\in\bR^K$ and $\bT$ is defined as 
\begin{equation}
\mathscr X(\mathbb T,X_0):=\Big\{\bm\xi=(\xi_1,\ldots,\xi_N)\,\Big|\,\text{$\bm\xi$ is admissible and }X_0+\sum_{k=1}^N \xi_k =0\text{ $\mathbb P$-a.s.}\Big\}.
\end{equation}
The set of deterministic liquidation strategies in $\mathscr X(\mathbb T,X_0)$ is denoted by $\mathscr X_{\textrm{det}}(\mathbb T,X_0)$. 
\end{definition}


We now turn toward the definition of the \emph{price impact} generated by an admissible strategy. As discussed in more detail in the introduction, in recent years several models were proposed that take the transience of price impact into account.  All these models, however, consider only one risky asset. 
In this paper, our goal is to extend the model from \citet{alfonsischiedslynko}, which is itself a linear and discrete-time version of the model from \citet{Gatheral}, to a situation with $K>1$ risky assets.   A  \emph{decay kernel} will be a continuous function
$$G:[0,\infty)\longrightarrow \bR^{K\times K}
$$
taking values in the space $\bR^{K\times K}$ of all real $K\times K$-matrices.
When  $\bm\xi$ is an admissible strategy for some time grid $\bT=\{t_1,\dots,t_N\}$ and $t\ge t_k\in\bT$, the  value $G_{ij}(t-t_k)$ describes the time-$t$ impact on the price of the $i^{\text{th}}$ asset  that was generated by trading one unit of the $j^{\text{th}}$ asset at time $t_k$. We therefore define the \emph{impacted price process} as 
\begin{equation}\label{eq-affectedPP}
S^{\bm\xi}_t = S^0_t + \sum_{t_k < t} G(t-t_k) \, \xi_{k},\qquad t\ge0.  
\end{equation}
Here $G(t-t_k) \, \xi_{k}$ denotes the application of the $K\times K$ matrix $G(t-t_k) $ to the $K$-dimensional vector $\xi_k$.

Let us write $S^{\bm\xi,i}_t$ for the  $i^{\text{th}}$ component of the price vector $S^{\bm\xi}_t=(S^{\bm\xi,1}_t,\dots,S^{\bm\xi,K}_t)^\top$. The execution of the $k^{\text{th}}$ order, $\xi_k$, shifts the price  of the $i^{\text{th}}$ asset linearly from $S^{\bm\xi,i}_{t_k}$ to 
$S^{\bm\xi,i}_{t_k+}$. The order $\xi_k^i$ of shares of the $i^{\text{th}}$ asset is therefore executed at the average price $\frac12(S^{\bm\xi,i}_{t_k+}+S^{\bm\xi,i}_{t_k})$.  The proceeds from executing the amount of $\xi_k^i$  shares of the $i^{\text{th}}$ asset are therefore given by $-\xi_k^i \frac12(S^{\bm\xi,i}_{t_k+}+S^{\bm\xi,i}_{t_k})$. It follows that the total revenues incurred by the strategy $\bm\xi$ are given by 
\begin{equation}\label{revenues def}
\mathscr R(\bm\xi)=-\frac12\sum_{k=1}^N\xi_k^\top  (S^{\bm\xi}_{t_k+}+S^{\bm\xi}_{t_k}). 
\end{equation}
In the sequel, it will be convenient to switch from revenues to costs, which are defined as the amount $X_0^\top S^0_0-\mathscr{R}(\bm\xi)$ by which the revenues fall short of the book value, $X_0^\top S_0^0$, of the initial portfolio.


\begin{remark}\label{order types remark}
In the one-dimensional version of our model, a bid-ask spread is often added so as to provide an interpretation of $\xi_k$ as a market order placed in a block-shaped limit order book; see, e.g., Section 2.6 in \citet{AS}. In practice, however, execution algorithms will use a variety of different order types, and one should think of price impact and  costs as being aggregated over these order types. For instance, while   half the spread has to be paid when placing a market buy order, the same amount can be earned when a limit sell order is executed. Other order types may  yield rebates when executed or may allow execution at mid price. So ignoring the bid-ask spread is probably more realistic than adding it to each single execution of an order. \end{remark}


 In this paper we will investigate the minimization of the \emph{expected costs} of a strategy, which in many situations is an appropriate optimization problem for determining optimal trade execution strategies. Our main interest, however, is to provide conditions on the decay kernel $G$ under which the model is sufficiently regular. As discussed at length in \citet{GatheralSchiedSurvey}, the regularity of a market impact model should be measured by the existence and behavior of execution strategies that minimize the expected costs, because the regularity of a model should be considered independently from the possible risk aversion that an agent using this model might have.
 
 To analyze the expected costs of an admissible strategy $\bm\xi=(\xi_1,\dots,\xi_N)$, it will be convenient to identify the particular realization, $\bm\xi(\omega)=(\xi_1(\omega),\dots,\xi_N(\omega))$, with an element of the tensor product space $\bR^N\otimes\bR^K$. We will also write $|\bT|$ for the cardinality of a time grid. 


\begin{lemma}\label{expected costs lemma}The expected costs of a strategy $\bm\xi \in \mathscr X(\mathbb T,X_0)$ for a time grid $\bT$ are given by 
\begin{equation}\label{expected costs formula}
\mathbb E[\,X_0^\top S_0-\mathscr{R}(\bm\xi)\,]=\mathbb E[\,C_\bT(\bm\xi)\,],
\end{equation}
where the \emph{cost function} $C_\bT: \mathbb R^{|\bT|} \otimes\bR^K\rightarrow \mathbb R$ is given by
\begin{equation}\label{C wt G representation}
C_\bT(\bm\xi) = \frac 12 \sum_{k,\ell=1}^N \xi_k^\top \wt G(t_k-t_\ell)\xi_\ell
\end{equation}
for the function $\wt G:\bR\to\bR^{K\times K}$ defined by 
\begin{equation}\label{tilde G def eq}
\wt G(t):=\begin{cases}G(t)&\text{for $t>0$,}\\
\frac12(G(0)^\top+G(0))&\text{for $t=0$,}\\
G(-t)^\top&\text{for $t<0$.}
\end{cases}
\end{equation}
\end{lemma}


We will now discuss the possible existence and structure of admissible strategies minimizing the expected costs within the class $\mathscr X(\mathbb T,X_0)$.   The problem of optimizing simultaneously over  time grids $\bT$ and strategies $\bm\xi\in\mathscr{X}(\bT,X_0)$ will be addressed in Section \ref{continuous-time section}.


\begin{lemma}\label{deterministic minimizers lemma}
There exists a strategy  in $\mathscr X(\mathbb T,X_0)$ that minimizes the expected costs $\mathbb E[\,C_\bT(\bm\eta)\,]$ among all strategies $\bm\eta\in \mathscr X(\mathbb T,X_0)$  if and only if there exists a deterministic strategy that minimizes the cost function $C_\bT(\bm\xi)$ over all $\bm\xi\in \mathscr X_{\text{\rm det}}(\mathbb T,X_0)$. In this case, any minimizer $\bm\eta^*\in\mathscr X(\mathbb T,X_0)$ can be regarded as a function from $\Omega$ into $\mathscr X_{\text{\rm det}}(\mathbb T,X_0)$ that takes $\mathbb P$-a.s.~values in the set of deterministic minimizers of the cost function $C_\bT(\cdot)$.  
\end{lemma}


  The condition  
  \begin{eqnarray}
  \mathbb E[\,C_\bT(\bm\eta)\,]\ge0 & \quad & \text{{for all  $\bT$, $X_0 \in \mathbb R^k$, and $\bm\eta \in \mathscr X(\mathbb T,X_0)$}} \label{regularity condition} 
  \end{eqnarray}
   can  be regarded as a regularity condition for the underlying market impact model. It rules out the possibility of obtaining positive expected profits through exploiting one's own price impact; see, e.g., \citet{alfonsischiedslynko} or \citet{GatheralSchiedSurvey} for  detailed discussions. In particular, it rules out the existence of \emph{price manipulation strategies} in the sense of \citet{hubermanstanzl}. In the sequel we will therefore focus on decay kernels that  satisfy \eqref{regularity condition}. It will turn out that  \eqref{regularity condition} can be equivalently characterized by requiring that the function $\wt G$ from \eqref{tilde G def eq} is a positive definite matrix-valued function in the following sense. 


\begin{definition}A function $H:\bR\to\bC^{K\times K}$ is called a \emph{positive definite matrix-valued function} if
for all $N\in\mathbb N$, $t_1,\dots, t_N\in\bR$, and $z_1,\dots,z_N\in\bC^K$,
\begin{equation}\label{wt G positive def eq}
\sum_{i,j=1}^Nz_i^*  H(t_i-t_j)z_j\ge0,
\end{equation}
where a $*$-superscript denotes the usual conjugate transpose of a complex vector or matrix. If moreover equality in \eqref{wt G positive def eq} can hold only for $z_1=\cdots=z_N=0$, then $H$ is called \emph{strictly positive definite}. When $K=1$, we say that $H$ is a \emph{(strictly) positive definite  function}.
\end{definition}


Note that a positive definite matrix-valued function $H$ is defined on the entire real line $\bR$ and is allowed to take values in the complex matrices. A decay kernel, $G$, on the other hand, is defined only on $[0,\infty)$ and takes values in the real matrices, $\bR^{K\times K}$. Considering the extended framework of $\bC^{K\times K}$-valued positive definite functions will turn out to be  convenient for our analysis.  The next proposition explains the relation between positive definite functions and decay kernels with nonnegative expected costs. 


\begin{proposition}\label{Cge0 pos def prop}
For a decay kernel $G$, the following conditions are equivalent.
\begin{enumerate}
\item $\mathbb E[\,C_\bT(\bm\eta)\,]\ge0$ for all time grids $\bT$, initial portfolios $X_0\in\bR^K$,  and    $\bm\eta\in\mathscr{X}(\bT,X_0)$.
\item $C_\bT(\bm\xi)\ge0$ for all time grids $\bT$ and $\bm\xi\in\bR^{|\bT|}\otimes\bR^K$.
\item For all time grids $\bT$, $C_\bT:\bR^{|\bT|}\otimes\bR^K\to\bR$ is convex. 
\item $\wt G$ defined in \eqref{tilde G def eq} is a positive definite matrix-valued  function.
\end{enumerate}
If moreover  these equivalent conditions are satisfied, then the equality $C_\bT(\bm\xi)=0$ holds for all time grids $\bT$ only for $\bm\xi=\bm0$, if and only if $\wt G$ is   strictly positive definite. In this case, $C_\bT:\bR^{|\bT|}\otimes\bR^K\to\bR$ is  strictly convex for all $\bT$.
\end{proposition}


Positive definiteness of $\wt G$ not only excludes the existence of price manipulation strategies. The following proposition states that it also guarantees the existence of strategies that minimize the expected  costs within a class $\mathscr{X}(\bT,X_0)$. Such strategies will be called \emph{optimal strategies} in the sequel. Once the existence of optimal strategies has been established, they can be computed by means of standard techniques from quadratic programming (see, e.g., \citet{Boot} or \citet{GillMurrayWright}).


\begin{proposition} \label{prop-opt-strat} Suppose that $\wt G$ is positive definite. Then  there exists an optimal strategy in $\mathscr{X}_{\text{\rm det}}(X_0,\bT)$ {\rm(}and hence in $\mathscr{X}(X_0,\bT)${\rm)} for all $X_0\in\bR^K$ and each time grid $\bT$.  Moreover, a strategy $\bm\xi\in \mathscr{X}_{\text{\rm det}}(X_0,\bT)$ is optimal if and only if there exists $\lambda\in\bR^K$ such that 
\begin{equation} \label{eq-lagrange-mult}
 \sum_{\ell=1}^N  \wt G(t_k-t_\ell)\xi_\ell=\lambda\qquad\text{for $k=1,\dots, |\bT|$}.
\end{equation}
If $\wt G$ is  strictly positive definite then  optimal strategies and the Lagrange multiplier $\lambda$ in \eqref{eq-lagrange-mult} are unique. 
\end{proposition}


Propositions~\ref{Cge0 pos def prop} and~\ref{prop-opt-strat} suggest that decay kernels $G$ for multivariate price impact should be constructed such that the corresponding function $\wt G$ from \eqref{tilde G def eq} is a positive definite matrix-valued function.  
Part (a) of the following elementary lemma implies that this can be achieved by defining $G(t):=H(t)$ for $t\ge0$ when $H:\bR\to\bR^{K\times K}$ is a given continuous positive definite matrix-valued function, because we will then automatically have $\wt G=H$. 


\begin{lemma}\label{Pos def funct Hermitian lemma}Let $H:\bR\to\bC^{K\times K}$ be a positive definite matrix-valued function. Then:
\begin{enumerate}
\item The matrix $H(0)$ is  nonnegative definite, and we have $H(-t)=H(t)^*$ for every $t\in\bR$.  In particular, $H(-t)=H(t)^\top$  if $H$ takes its values in $\bR^{K\times K}$.  
\item Also $t\mapsto H(t)^*$ is  a positive definite matrix-valued function; it is   strictly positive definite if and only $H$ is strictly positive definite.
\end{enumerate}
\end{lemma}


 Due to the established one-to-one correspondence of decay kernels with nonnegative expected costs and continuous $\bR^{K\times K}$-valued positive definite functions, we will henceforth use the following terminology.


\begin{definition}\label{positive definite kernel def} A decay kernel $G:[0,\infty)\to\bR^{K\times K}$ is called \emph{(strictly) positive definite} if the corresponding function $\wt G$ from \eqref{tilde G def eq} is a (strictly) positive definite matrix-valued function.
\end{definition}


\subsection{Integral representation of positive definite decay kernels}

We turn now to characterizations of the positive definiteness of a matrix-valued function. 
In the one-dimensional situation,  $K=1$,  Bochner's theorem \citep{bochner32} characterizes all  continuous positive definite functions as the Fourier transforms of nonnegative finite Borel measures.  There are several extensions of Bochner's theorem  to the case of matrix-  or operator-valued functions. Some of these results will be combined in Theorem~\ref{bochner-thm-nonsymm} and Corollary~\ref{Symmetric Bochner Cor} below. For the corresponding statements, we first introduce some terminology. 

As usual, a complex matrix $N\in\mathbb C^{n\times n}$  is called \emph{nonnegative definite} if $z^* N z\ge 0$ for every $z\in\mathbb C^n$. When even $z^* N z>0$ for every nonzero $z$, $N$ is called \emph{strictly positive definite}.  A nonnegative definite complex matrix $N\in\mathbb C^{n\times n}$ is necessarily Hermitian, i.e. $N=N^*$. In particular, a real matrix $N\in \mathbb R^{n\times n}$ is nonnegative definite if and only if it belongs to the set $\bS_+(n)$ of nonnegative definite symmetric real $n\times n$-matrices.  By $\bS(n)$ we denote the set of all symmetric matrices in $\bR^{n\times n}$. 
An arbitrary real matrix $M\in\mathbb R^{n\times n}$ will be called \emph{nonnegative} if $x^\top M x \ge 0$ for every $x\in\mathbb R^n$ and \emph{strictly positive} if $x^\top Mx>0$ for all nonzero $x\in\bR^n$. Note that a real matrix $M\in\mathbb R^{n\times n}$ is nonnegative if and only if its symmetric part,  $\frac 12(M+M^\top)$, is nonnegative definite.

Let $\mathscr B(\mathbb R)$ be the Borel $\sigma$-algebra on $\mathbb R$.
 A mapping $M:\mathscr B(\mathbb R)\rightarrow \mathbb C^{K\times K}$ will be called a \emph{nonnegative definite matrix-valued measure} if  every component $M_{ij}$ is a complex measure with finite total variation and the matrix $M(A)\in \mathbb C^{K\times K}$ is nonnegative definite for every $A\in\mathscr B(\mathbb R)$.

The following theorem combines  results by  \citet{cramer40}, \citet{falb69}, and \citet{naimark43};  we refer to \citet{Glockner} for extensions of this result  and for a comprehensive historical account.


\begin{theorem} \label{bochner-thm-nonsymm} For a continuous function $H:\bR\to\bC^{K\times K}$  the following are equivalent.
\begin{enumerate}
\item $H$ is a positive definite matrix-valued function.
\item  For every $z\in\bC^K$, the complex function 
$  t\mapsto z^*H(t)z$
is positive definite.
\item $H$ is the Fourier transform of a  nonnegative definite matrix-valued measure $M$, i.e.,
\begin{equation} \label{eq-ComplexFourierTrafo}
H(t)=\int_{\mathbb R} e^{i\gamma t} \,M(d\gamma)\qquad\text{for $t\in\bR$. }
\end{equation}
\end{enumerate}
Moreover, any matrix-valued measure $M$ with \eqref{eq-ComplexFourierTrafo} is uniquely determined by $H$.
\end{theorem}

\begin{proof} The equivalence of (b) and (c) was proved in \citet{falb69}. The equivalence of (a) and (c) follows from two statements in the book by  \citet{gihmanskorohod}, namely the remark after Theorem 1 in \S1 of Chapter IV and  Theorem 5 in \S2 of Chapter IV. The uniqueness of $M$ is standard.
\end{proof}


The preceding theorem simplifies as follows when considering positive definite functions $H$ taking values in the space $\bS(K)$ of symmetric real $K\times K$-matrices. By Lemma~\ref{Pos def funct Hermitian lemma} (a), such functions $H$ correspond to positive definite decay kernels  $G$ that are \emph{symmetric} in the sense that $G(t)^\top=G(t)$ for all $t\ge0$.  In this case, we have $H(t)=\wt G(t)=G(|t|)$ for all $t\in\bR$.


\begin{corollary}\label{Symmetric Bochner Cor}For a continuous function $H:\bR\to \bC^{K\times K}$ the following statements are equivalent.
\begin{enumerate}
\item $H(t)\in\bS(K)$ for all $t$, and $ H$ is a positive definite matrix-valued function.
\item $H(t)\in\bS(K)$ for all $t$, and the real  function $  t\mapsto x^\top H(t)x$ is positive definite  for every $x\in\bR^K$.
\item $H$ admits a representation \eqref{eq-ComplexFourierTrafo} with a nonnegative definite measure    $M$ that takes values in $\bR^{K\times K}$ (and hence in $\bS_+(K)$) and is symmetric on $\bR$ in the sense that $M(A)= M(-A)$ for all  $A\in\mathscr{B}(\bR)$.
\end{enumerate}
\end{corollary}

\medskip

\begin{remark}[\bfseries Discontinuous positive definite functions and temporary price impact]\label{GlocknerRemark}Let $H_0$ be a nonzero nonnegative definite matrix. Then $H(t):=H_0\indf{\{0\}}(t)$ is a positive definite matrix-valued function that is not continuous and therefore does not admit a representation \eqref{eq-ComplexFourierTrafo}. It is possible, however, to give a similar  integral representation also for  discontinuous  matrix-valued positive definite functions satisfying a certain boundedness condition. To this end, one needs to replace the measure $M$ by a nonnegative definite matrix-valued measure on the larger space of characters for the additive (semi-)groups $\bR$ or $\bR_+$; see \citet[Theorem 15.7]{Glockner}. In the context of price impact modeling, the costs \eqref{C wt G representation} associated with a discontinuous decay kernel of the form $G(t):=G_0\indf{\{0\}}(t)$ for some nonnegative matrix $G_0$ can be viewed as resulting from \emph{temporary price impact} that affects only the order that has triggered it and disappears immediately afterwards; see \citet{BertsimasLo} and \citet{almgrenchriss2001} for temporary price impact in   one-dimensional models. More generally, to take account the discontinuity $G(0)-G(0+) \in \bS_+(K)$, one will have to precise the definition \eqref{revenues def} of the revenues by assuming $\mathscr R(\bm\xi)=- \sum_{k=1}^N\xi_k^\top  (S^{\bm\xi}_{t_k}+\frac12 G(0)\xi_k)$ (note that this is $G(0)$ and not $G(0+)$). Last, let us mention that the discontinuity at $0$ is the only one relevant in practice: other discontinuities would generate a weird and predictable price impact. Thus, the temporary price impact can be handled separately and assuming $G$ continuous is not restrictive.\end{remark}

\subsection{Convex, nonincreasing, and nonnegative decay kernels}

As shown and discussed in \citet{alfonsischiedslynko}, not every decay kernel $G:[0,\infty)\to\bR$ with positive definite $\wt G$ is a reasonable model for the decay of price impact in a single-asset model. Specifically  it was shown  that for $K=1$ it makes sense to require that decay kernels are nonnegative, nonincreasing, and convex. Since similar effects as in \citet{alfonsischiedslynko} can also be observed  in our multivariate setting (see Figure~\ref{oscillations-fig}), we  need to introduce and analyze further conditions to be satisfied by $G$. 
 To motivate the following definition, consider two trades $\xi_1$ and $\xi_2$ placed at times $t_1 < t_2$. The quantity $\xi_2^\top G(t_2-t_1) \xi_1$ describes that part of the liquidation costs for the order $\xi_2$ that was caused by the order $\xi_1$.  When $\xi_1=\xi_2$, it is intuitively clear that these costs should be nonnegative and nonincreasing in $t_2-t_1$.


\begin{definition}\label{Property def}
A matrix-valued function $G:[0,\infty)\to\mathbb R^{K\times K}$ is called
\begin{enumerate}
\item \emph{nonincreasing}, if for every $ x \in \mathbb R^K$ the function $t\mapsto x^\top G(t) x$ is nonincreasing;
\item \emph{nonnegative}, if     $G(t)$ is a nonnegative matrix for every  $t \in [0,\infty)$;
\item \emph{(strictly) convex}, if for all $ x\in\mathbb R^K$ the function $ t\mapsto x^\top G(t) x$  is (strictly) convex.
\end{enumerate}
\end{definition}


Here and in Lemma~\ref{prop-nonneg} and Theorem~\ref{convex strict pd} below, we do not assume that $G$ is continuous. Note that the properties introduced in the preceding definition depend only on the symmetrization, $\frac12(G^\top+G)$, of $G$. We have the following simple result on two properties introduced in Definition~\ref{Property def}.


\begin{lemma}\label{prop-nonneg} Suppose that $G:[0,\infty)\to \bR^{K\times K}$ is a nonincreasing and positive definite decay kernel. Then $G$ is nonnegative.\end{lemma}


If $G$ is nonincreasing, nonnegative, and convex, then so is the function $g^ x(t):= x^\top G(t) x$ for each $ x\in\mathbb{R}^K$. Hence, $t\mapsto g^ x(|t|)$ is a positive definite function due to a criterion often attributed to \citet{Polya}, although this criterion is also an easy consequence of  \citet{Young}. It hence follows from Corollary~\ref{Symmetric Bochner Cor}  that also the matrix-valued function $\wt G$ is  positive definite as soon as $G$ is symmetric and continuous.  But  an even stronger result is possible: $G$ is even \emph{strictly} positive definite as soon as $g^ x$ is nonincreasing, nonnegative,  convex, and nonconstant for each nonzero $ x\in\mathbb R^K$. This is the content of our subsequent theorem, which extends the corresponding result for $K=1$  (see Theorems 3.9.11 and  3.1.6 in \citet{Sasvari} or Proposition 2 in  \citet{alfonsischiedslynko} for two different proofs) and is of independent interest.


\begin{theorem}\label{convex strict pd} If $G:[0,\infty)\to\bR^{K\times K}$ is symmetric, nonnegative, nonincreasing, and convex then $G$ is positive definite. Moreover, $G$ is even strictly positive definite if and only if $t\mapsto x^\top G(t) x$ is nonconstant  for each nonzero $ x\in\mathbb R^K$.\end{theorem}


We will see in Proposition~\ref{prop-convex-nonposdef} that in Theorem~\ref{convex strict pd} we can typically not dispense of the requirement that $G$ is symmetric to conclude   positive definiteness.

\subsection{Commuting decay kernels}

We will now introduce another property that one can require from a decay kernel.


\begin{definition}A decay kernel $G:[0,\infty)\to\mathbb{R}^{K\times K}$ is called \emph{commuting} if $G(t)G(s)=G(s)G(t)$ holds for all $s,t\ge0$.  
\end{definition}


If a symmetric decay kernel is commuting, it may be simultaneously diagonalized,  and its properties can be characterized via the resulting collection of one-dimensional decay kernels, as  explained in the following proposition.

\begin{proposition}\label{commuting properties Prop}A symmetric decay kernel $G$ is commuting if and only if there exists an orthogonal matrix $O$ and  functions $g_1,\dots,g_K:[0,\infty)\to\mathbb R$ such that 
\begin{equation}\label{simultaneous diagonalization eq}
G(t)=O^\top\text{\rm diag}(g_1(t),\dots,g_K(t))\,O.
\end{equation}
Moreover, the following assertions hold.
\begin{enumerate}
\item $G$ is (strictly) positive definite if and only if the $\bR$-valued  functions $t\mapsto g_i(t)$ are (strictly) positive definite for all $i$.
\item  $G$ is nonnegative if and only if $g_i(t)\ge 0$ for all $i$ and $t$.
\item $G$ is nonincreasing if and only if $g_i$ is nonincreasing for all $i$.
\item $G$ is convex if and only if $g_i$ is convex  for all $i$.
\item If $G$ is positive definite, then a strategy $\bm\xi=(\xi_1,\dots,\xi_{|\bT|})\in\mathscr{X}_{\text{\rm det}}(\bT,X_0)$ is optimal if and only if it is of the form 
$$\xi_j=O^\top(\eta_j^1,\dots,\eta_j^K)^\top,
$$
where $\bm\eta^i=(\eta_1^i,\dots,\eta^i_{|\bT|})\in\bR^{|\bT|}\otimes\bR$ is an optimal strategy in $\mathscr{X}_{\text{\rm det}}(\bT,(OX_0)^i)$ for the one-dimensional decay kernel $g_i$ (here $(OX_0)^i$ denotes the $i^{\text{th}}$ component of the vector $OX_0$). 
\end{enumerate} 
\end{proposition}


For $K=1$, we know that a nonnegative nonincreasing convex function is positive definite, and even strictly positive definite when it is nonconstant. Thus,   Proposition~\ref{commuting properties Prop} implies Theorem~\ref{convex strict pd} in the special situation of commuting decay kernels.

In the case $K=1$, \citet{alfonsischiedslynko} observed that there exist nonincreasing, nonnegative, and strictly positive definite decay kernels $G$ for which the optimal strategies exhibit strong oscillations between buy and sell orders (\lq\lq transaction-triggered price manipulation"); see  Figure~\ref{oscillations-fig} for an example in our multivariate setting.
 Theorem 1 in \citet{alfonsischiedslynko} gives conditions that exclude such oscillatory strategies for $K=1$ and guarantee that optimal strategies are buy-only or sell-only: $G$ should be nonnegative, nonincreasing, and convex.  
For $K>1$, however, the situation changes and one cannot expect to exclude the coexistence of buy and sell orders in the same asset. The reason is that liquidating a position in a first  asset may create a drift in the price of a second asset through cross-asset price impact. Exploiting this drift in the second asset via a  round trip may help to mitigate the costs resulting from liquidating the position in the first asset; see Figure~\ref{fig-1}. Therefore one cannot hope to completely rule out all round trips for decay kernels that are not diagonal. Nevertheless, our next result gives conditions on $G$ under which optimal strategies can be expressed as linear combinations of $K$ strategies with buy-only/sell-only components which leads to a uniform bound of the total number of shares traded by the optimal strategy, preventing large oscillations as in Figure~\ref{oscillations-fig}. This result will also allow us to construct minimizers on non-discrete time grids in Section \ref{continuous-time section} below. 
\newsavebox{\smlmat}
\savebox{\smlmat}{$\left(\begin{smallmatrix}1&\rho\\\rho&1\end{smallmatrix}\right)$}
\begin{figure}
\begin{center}
\includegraphics[width=7cm]{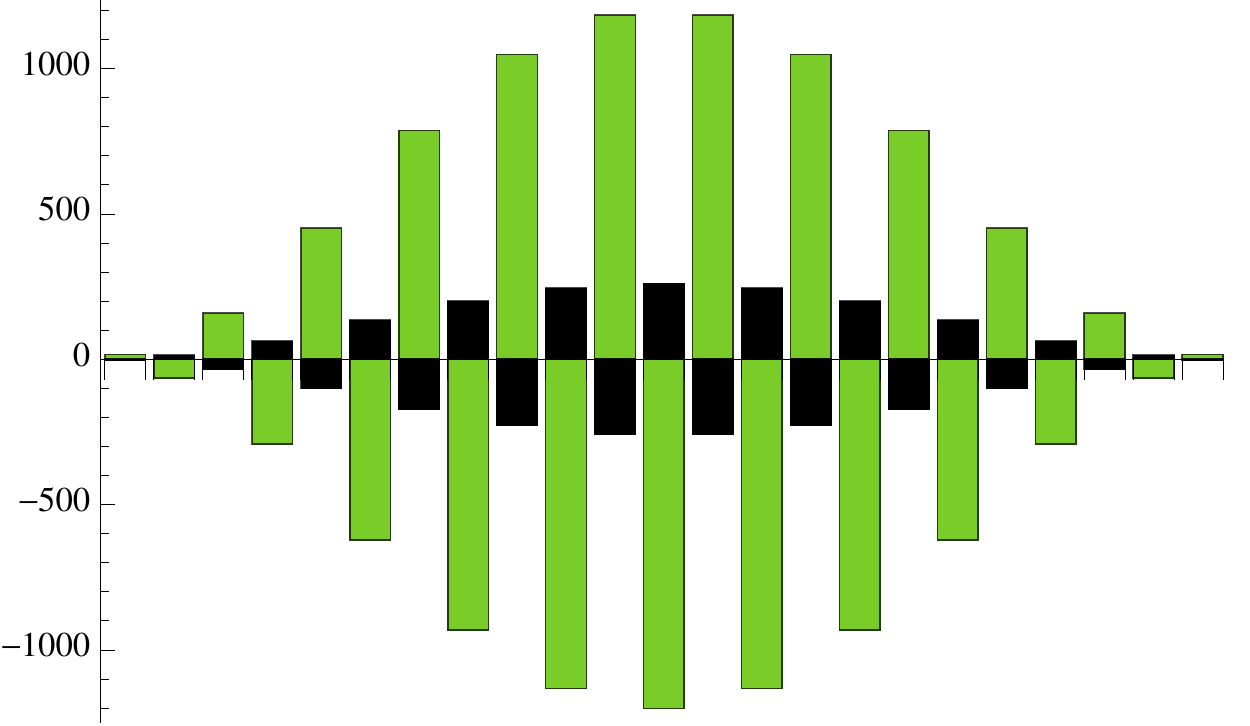}
\caption{Optimal strategy $\bm\xi$ for $X_0=(10,0)$, $N=23$, and the  strictly positive definite decay kernel   $G(t)=\exp(-(tB)^2)$ for~$B=\usebox{\smlmat}$. The first component of $\bm\xi$ is plotted in green, the second component in black. Note that the amplitude of the oscillations exceeds the initial asset position by a factor of more than 110. That $G$ is strictly positive definite  follows from Remark \ref{matrix functions remark}.}\label{oscillations-fig}
\end{center}
\end{figure}


\begin{proposition}\label{prop-orthonormalbasis} Let $G$ be a symmetric, nonnegative, nonincreasing, convex and commuting decay kernel. 
Then  there exist an orthonormal basis $v_1,\dots, v_K$ of $\mathbb R^K$ and, for each time grid $\bT$, optimal strategies $\bm\xi^{(i)}\in\mathscr{X}_{\text{\rm det}}(\bT,v_i)$, $i=1,\dots,K$, such that the following conditions hold.
\begin{enumerate}
\item The components of each $\bm\xi^{(i)}$ consist of buy-only or sell-only strategies. More precisely, for $i,j\in\{1,\dots,K\}$ and $n,m\in\{1,\dots,|\bT|\}$ we have $\xi^{(i),j}_m\xi^{(i),j}_n\ge0$. 
\item For $X_0=\sum_{i=1}^K\alpha_iv_i\in\bR^K$ given, $\bm\xi:=\sum_{i=1}^K\alpha_i\bm\xi^{(i)}$ is an optimal strategy in $\mathscr{X}_{\text{\rm det}}(\bT,X_0)$. 
\end{enumerate}
\end{proposition}


Note that for $K=1$ every decay kernel is symmetric and commuting.  Hence, for $K=1$ the preceding proposition reduces to Theorem~1 in \citet{alfonsischiedslynko}: The optimal strategy for a one-dimensional, nonconstant, nonnegative, nonincreasing, and convex decay kernel is buy-only or sell-only.


\begin{remark}Oscillations of trading strategies as those observed in Figure \ref{oscillations-fig} can be prevented by adding sufficiently high transaction costs to each trade. Such transaction costs arise naturally if only market orders are permitted; see, e.g.,  Sections 7.1 and 7.2 in \citet{BussetiLillo}. As discussed in Remark \ref{order types remark}, however, actual trading strategies will often incur much lower transaction costs than strategies that only use market orders and, if  transaction costs are sufficiently small, oscillations  may only be dampened but not be completely eliminated. As a matter of fact, oscillatory trading strategies of high-frequency traders played a major role in the \lq\lq Flash Crash" of May 6, 2010; see \citet[p. 3]{SEC}. \end{remark}


 Propositions~\ref{commuting properties Prop}  and~\ref{prop-orthonormalbasis} give not only a characterization of nice properties of certain  decay kernels. They also provide a way of constructing decay kernels $G$ that have all desirable properties. One simply needs to start with an orthogonal matrix $O$ and nonincreasing, convex, and nonconstant functions $g_1,\dots,g_K:[0,\infty)\to[0,\infty)$ and then define a decay kernel as $G(t)=O^\top\text{\rm diag}(g_1(t),\dots,g_K(t))\,O$. A special case of this construction is provided by the so-called \emph{matrix functions}, which we will explain in the sequel; see also Section~\ref{matrix fcts section} for several examples in this context.

Let $g:[0,\infty)\to\mathbb R$ be a function and $B\in \bS_+(K)$. Then there exists an orthogonal matrix $O$ such that
$B=O^\top\text{diag}(\rho_1,\dots,\rho_K)O$, where $\rho_1,\dots,\rho_K\ge0$ are the eigenvalues of $B$. The matrix $g(B)\in \bS(K)$ is then defined as
\begin{equation}\label{matrix function def eq}
g(B):=O^\top \text{diag}(g(\rho_1),\dots,g(\rho_K))O;
\end{equation}
see, e.g., \citet{Donoghue}.
We can thus define a decay kernel $G:[0,\infty)\to \bS(K)$ by
\begin{equation}\label{matrix function eq}
G(t)=g(tB)=O^\top \text{diag}(g(t\rho_1),\dots,g(t\rho_K))O,\qquad t\ge0.
\end{equation}
We summarize the properties of $G$ in the following remark. In Section~\ref{matrix fcts section} we will analyze decay kernels that arise as matrix exponentials and explicitly compute the corresponding optimal strategies.

\begin{remark}\label{matrix functions remark}The decay kernel $G$ defined in \eqref{matrix function eq}
 is commuting.  Moreover, it is of the form \eqref{simultaneous diagonalization eq} with $g_i(t)=g(t\rho_i)$, and so Proposition~\ref{commuting properties Prop} characterizes the properties of $G$. In particular, it is positive definite if and only if $t\mapsto g(|t|)$ is a positive definite function. Moreover, it will be nonnegative, nonincreasing, or convex if and only if $g$ has the corresponding properties. In addition, optimal strategies can be computed via Proposition~\ref{commuting properties Prop} (e). 
\end{remark}


\begin{remark}
Let $G$ be a  nonnegative, symmetric and commuting decay kernel, and $f:[0,\infty)\to[0,\infty) $ be a convex nondecreasing function. We define the kernel~$F$ by
$F(t):=f(G(t))$ as in \eqref{matrix function def eq}.  We get easily from Proposition~\ref{commuting properties Prop} that $F$ is  nonincreasing and convex if $G$ is also  nonincreasing and convex. It is therefore positive definite in this case. 
\end{remark}

   \subsection{Strategies on non-discrete time grids}\label{continuous-time section}

If $\bT$ and $\bT'$ are time grids such that $\bT\subset\bT'$, then $\mathscr{X}(\mathbb T,X_0)\subset\mathscr{X}(\mathbb T',X_0)$ and hence 
$$\min_{\bm\xi\in\mathscr{X}(\mathbb T,X_0)}\mathbb{E}[\,C_{\mathbb T}(\bm\xi)\,]\ge \min_{\bm\xi'\in\mathscr{X}(\mathbb T',X_0)}\mathbb{E}[\,C_{\mathbb T'}(\bm\xi')\,].
$$
 It is therefore clear that problem of minimizing $\mathbb{E}[\,C_{\mathbb T}(\bm\xi)\,]$ jointly over $\bm\xi\in\mathscr{X}(\mathbb T,X_0)$ and time grids $\bT$ has in general no solution within the class of finite time grids. For this reason it is natural to consider an extension of our framework to non-discrete time grids. For the one-dimensional case $K=1$ a corresponding framework was developed in \citet{gatheralschiedslynko}. Proposition \ref{prop-orthonormalbasis} (b) will enable us to obtain a similar extension in our present framework.

\medskip

\begin{definition}\label{cont time def}Let $\bT$ be an arbitrary compact subset of $[0,T]$. An \emph{admissible strategy for $\bT$} is a left-continuous, adapted, and bounded $K$-dimensional stochastic process $( X_t)$ such that $t\mapsto X^i_t$ is of finite variation for $i=1,\dots, K$ and satisfies $  X_t=0$ for all $t>T$. We assume furthermore that the vector-valued random measure $d X_t$ is supported on $\bT$ and that its components have $\mathbb{P}$-a.s.~bounded total variation. The class of strategies with given initial condition $X_0$ will be denoted by $\mathscr{X}(\bT,X_0)$, the subset of deterministic strategies in $\mathscr{X}(\bT,X_0)$ will be denoted by $\mathscr{X}_{\text{det}}(\bT,X_0)$.
\end{definition}

If $\bT=\{t_1,\dots,t_N\}$ is a finite time grid and $\bm\xi\in\mathscr{X}(\bT,X_0)$ is an admissible strategy in the sense of Definition \ref{discrete-time def}, then  
\begin{align}\label{X xi}
X^{\bm\xi}_t:=X_0-\sum_{t_k\in\bT,\,t_k<t}\xi_k
\end{align}
is an admissible strategy in the sense of Definition \ref{cont time def}. Therefore Definition \ref{cont time def} is consistent with Definition \ref{discrete-time def}. 
Now let $\bT$ be an arbitrary compact subset of $[0,T]$ and $G$ be a decay kernel. For $X\in\mathscr{X}(\bT,X_0)$ we define the  associated costs as 
$$C_{\bT}(X):=\frac12\int_{\bT}\bigg(\int_{\bT}\wt G(t-s)\,dX_s\bigg)^\top\,dX_t.
$$
When $\bT$ is a finite time grid, $\bm\xi\in\mathscr{X}(\bT,X_0)$, and $X^{\bm\xi}$ is defined by \eqref{X xi} then we clearly have $C_{\bT}(X^{\bm\xi})=C_{\bT}(\bm\xi)$, and so also the  definition of the cost functional is  consistent with our earlier definition for discrete time grids. We have the following result. 

\begin{theorem}\label{cont thm}Let $G$ be a symmetric, nonnegative, nonincreasing, convex, nonconstant, and commuting decay kernel and $\bT$ be a compact subset of $[0,T]$. Then the following assertions hold.
\begin{enumerate}
\item For $X_0\in\mathbb{R}^K$ there exists precisely one strategy $X^*\in\mathscr{X}(\bT,X_0)$ that minimizes the expected costs, $\mathbb{E}[\,C_{\bT}(X)\,]$, over all strategies  $X\in\mathscr{X}(\bT,X_0)$. Moreover, $X^*$ is deterministic and can be characterized as the unique strategy in $\mathscr{X}_{\text{det}}(\bT,X_0)$ that solves the following generalized Fredholm integral equation for some $\lambda\in\bR^K$,
\begin{align}\label{Fredholm integral eq}
\int_{\bT}\wt G(t-s)\,dX_s=\lambda\qquad\text{for all $t\in\bT$.}
\end{align}
\item Let $\mathscr{T}$ denote the class of all finite time grids in $\bT$. Then
$$\inf_{\bT'\in\mathscr{T}}\min_{\bm\xi\in\mathscr{X}(\bT',X_0)}\mathbb{E}[\,C_{\bT'}(\bm\xi)\,]=\mathbb{E}[\,C_{\bT}(X^*)\,].
$$
\end{enumerate}
\end{theorem}

\section{Examples}
\label{section-4}

\subsection{Constructing decay kernels by transformation}

In this section  we will now look at some transformations of decay kernels. The first of these results concerns decay kernels of the simple form $G(t)=g(t)L$ where $g:[0,\infty)\to\bR$ is a function and $L\in\bR^{K\times K}$ is a fixed matrix.


\begin{proposition}\label{prop-decomposable-G}
 For $L\in\bS_+(K)$  and a  positive definite function $g:\bR\rightarrow\mathbb R$, the decay kernel $  G(t):=g(t)L$ is   positive definite.
 If, moreover,  $g$ is a strictly positive definite function and $L$ is a strictly positive definite matrix, then $  G$ is also strictly positive.
\end{proposition}


The simple decay kernels  from the preceding proposition provide a class of examples to which also the next result applies. In particular, by choosing in the subsequent Proposition~\ref{prop-L-transl} the decay kernel as $G(t):=g(t)\text{Id}$ for $g:\bR\rightarrow\mathbb R$ positive definite and $\text{Id}\in\bR^{K\times K}$ denoting the identity matrix,  one sees that the optimal strategies for  decay kernels of the form $g(t)L$ with $L\in\bS_+(K)$ do not depend on the cross-asset impact $g(t)L_{ij}$ for $i\neq j$.  Hence, cross-asset impact will only become   relevant when the components of $G$ decay at varying rates.  


\begin{proposition} \label{prop-L-transl} Let $G$ be a decay kernel and define $G_L(t):=LG(t)$ for some $L\in\bR^{K\times K}$. When both $  G$ and $  G_L$ are positive definite, then every optimal strategy in  $\mathscr{X}_{\text{\rm det}}(\bT,X_0)$  for $G$ is also an optimal strategy for $G_L$.  
\end{proposition}


The main message obtained from combining Propositions \ref{prop-decomposable-G} and \ref{prop-L-transl}  is the following: if the price impact between all pairs of assets decays at the same rate, then cross-asset impact can be ignored and one can simply consider each asset individually.

We show next that also  congruence transforms preserve positive definiteness. This result extends Proposition \ref{commuting properties Prop} (e).


\begin{proposition} \label{prop-L-transform}
If $G$ is a (strictly) positive definite decay kernel and $L$ and an invertible $K\times K$ matrix, then $G^L:=L^\top G(t)L$ is (strictly) positive definite. If, moreover, $\bm\xi$ is an optimal strategy for $G$ in $ \mathscr{X}_{\text{\rm det}}(\bT,LX_0)$, then $\bm\xi^L:=(L^{-1}\xi_1,\dots,L^{-1}\xi_{|\bT|})$ is an optimal strategy for $G^L$ in $ \mathscr{X}_{\text{\rm det}}(\bT,X_0)$.
\end{proposition}

 \begin{example}[\bfseries Permanent impact] Let $G(t)=G_0$, where $G_0$ is any fixed matrix in $\mathbb{R}^{K\times K}$. For any time grid $\bT$, $X_0\in\mathbb{R}^K$, and $\bm\xi\in\mathscr{X}(\bT,X_0)$  we then have $C_{\bT}(\bm\xi)=X_0^\top G_0X_0$. Hence $G$ is positive definite as soon as $G_0$ is nonnegative. By taking $G_0$ such that $ X_0^\top G_0X_0\ge0$ for some nonzero $X_0$ and $Y_0^\top G_0Y_0<0$ for some other $Y_0$ one gets an example illustrating that it is not possible to fix $X_0$ in part (a) of Proposition \ref{Cge0 pos def prop}.\end{example}

\subsection{Exponential decay kernels }\label{matrix fcts section}

In this section we will  discuss decay kernels with an exponential decay of price impact. For $K=1$ exponential decay was introduced in  \citet{ObizhaevaWang} and further studied, e.g., in \citet{AFS1} and \citet{PredoiuShaikhetShreve}. The next example extends the results from  \citet{ObizhaevaWang} and  \citet{AFS1} to a multivariate setting in which the decay kernel is defined in terms of matrix exponentials. The remaining results of this section are stated in a more general but two-dimensional context. The main message of these examples is that, on the one hand, it  is easy to construct decay kernels with all desirable properties via matrix functions. But, on the other hand, it is typically not easy to establish properties such as positive definiteness for decay kernels that are defined coordinate-wise.


\begin{example}[\bfseries  Matrix exponentials]\label{example-matrixexp}  For an orthogonal matrix $O$, $\rho_1,\dots,\rho_K\ge0$, and $B=O^\top\text{diag}(\rho_1,\dots,\rho_K)O\in \bS_+(K)$, the decay kernel $G(t)=\exp(-tB)$
is of the form \eqref{matrix function eq} with $g(t)=e^{-t}$. It follows that $G$ is nonnegative, nonincreasing, and convex. In particular, $G$ is  positive definite. When the matrix $B$ is strictly positive definite, as we will assume from now on, the decay kernel $G$ is even  strictly positive definite. We now  compute the optimal strategy $\bm\xi=(\xi_1,\dots,\xi_N)$ for an initial portfolio $X_0\in\mathbb R^K$ and time grid $\bT=\{t_1,\dots,t_N\}$. To this end, we will use part (e) of Proposition~\ref{commuting properties Prop}. Let $\bm\eta^i:=(\eta^i_1,\dots, \eta^i_N)$ be the optimal strategy for the initial position $y_i$ and for the one-dimensional decay kernel  $g_i(t)=e^{-t\rho_i}$.  Let 
$$a_n^i:=e^{-(t_n-t_{n-1})\rho_i}\qquad\text{and}\qquad \lambda_i:=\frac{ -y_i}{\frac{2}{1+a^i_2} + \sum_{n=3}^N \frac{1-a^i_n}{1+a^i_n} }.
$$
Theorem 3.1 in \citet{AFS1} implies that
 the optimal strategy $\bm\eta^i=(\eta^i_1,\dots, \eta^i_N)$ in $\mathscr{X}_{\text{get}}(\bT,y^i)$ is given by
$$\eta^i_1=\frac{\lambda_i}{1+a^i_2},\quad \eta^i_n=\Big(\frac{1}{1+a^i_{n}} -\frac{a^i_{n+1}}{1+a^i_{n+1}}\Big)\lambda_i\text{ for }
n=2,\dots,N-1,\quad\text{and}\quad
\eta^i_N= \frac{\lambda_i}{1+a^i_N} .
$$
Via part (e) of Proposition~\ref{commuting properties Prop}, we can now compute the optimal strategy $\bm\xi$. Consider first the optimal strategy $\bm{\eta}$ for the decay kernel $D(t):=\text{diag}(\exp(-\rho_1t), \ldots,\exp(-\rho_Kt))$ and initial position $O X_0$. Then $\bm\eta=(\bm\eta^1,\dots,\bm\eta^K)^\top$ for $y^i:=(OX_0)^i$. When defining $Q_n:=D(t_n-t_{n-1})$ and
$$\wt\lambda=-\left(2(\IdM+Q_2)^{-1}+\sum_{n=3}^N (\IdM-Q_n)(\IdM+Q_n)^{-1}\right)^{-1}O X_0,$$
 ${\bm{\eta}}=(\eta_1,\dots,\eta_N)$ can be conveniently expressed as follows: 
 \begin{eqnarray*}
 \eta_1&=&(1+Q_2)^{-1}\wt\lambda,\\
 \eta_n&=&(\IdM+Q_n)^{-1}\wt\lambda-Q_{n+1}(\IdM+Q_{n+1})^{-1}\wt\lambda \quad\text{for $n=2,\ldots,N-1$,}\\
  \eta_N&=&(\IdM+Q_N)^{-1}\wt\lambda.
\end{eqnarray*}
By  part (e) of Proposition~\ref{commuting properties Prop} the optimal strategy $\bm\xi$ for $G$ and $X_0$ is now given by $\bm\xi=O^T \eta$.
To remove $O$ from these expressions, define $A_n=e^{-(t_n-t_{n-1})B}=O^\top Q_n O$ and
$$\lambda:=-\bigg[2\big(\text{Id}+A_2\big)^{-1}+\sum_{i=3}^N\big(\text{Id}-A_i\big)\big(\text{Id}+A_i\big)^{-1}\bigg]^{-1}X_0.
$$
By observing that  $(\IdM+A_n)^{-1}=O^\top(\IdM+Q_n)^{-1}O$ and  $\lambda =O^\top \wt\lambda$, we find that the components of the optimal strategy $\bm\xi$ are
\begin{align*}\xi_1&=\big(\text{Id}+A_2\big)^{-1}\lambda,\\
\xi_n&=\big(\text{Id}+A_n\big)^{-1}\lambda -A_{n+1}\big(\text{Id}+A_{n+1}\big)^{-1}\lambda\quad\text{for $n=2,\dots, N-1$,}\\
\xi_N&= \big(\text{Id}+A_N\big)^{-1}\lambda.
\end{align*}
Let us finally consider the situation of an equidistant time grid, $t_i=\frac{i-1}{N-1}$. In this case, all matrices $A_i$ are equal to a single matrix $A$. Our formula for $\lambda$ then becomes
$$\lambda=-(\IdM+A)\Big(N\IdM-(N-2)A\Big)^{-1}X_0.
$$
The formula for the optimal strategy thus simplifies to
\begin{eqnarray*}
\xi_1&=&-\Big(N\IdM-(N-2)A\Big)^{-1}X_0,\\
\xi_i&=&(\text{Id}-A)\xi_1\quad \text{for }i=2,\dots, N-1,\\
\xi_N&=&\xi_1.
\end{eqnarray*}
It is not difficult to extend this result to the setting of Section \ref{continuous-time section} by arguing as in \citet[Example 2.12]{gatheralschiedslynko}. The details are left to the reader.\hfill$\diamondsuit$
\end{example}


\newsavebox{\matB}
\savebox{\matB}{$\left(\begin{smallmatrix}b&1\\0&b\end{smallmatrix}\right)$}
When $g:\bR\to\bR$ is an analytic function, the definition of $g(B)$ is also possible for nonsymmetric matrices by letting
$$g(B):=\sum_{k=0}^\infty a_kB^k,
$$
where $g(x)=\sum_{k=0}^\infty a_kx^k$ is the power series development of $g$. In the following example we analyze the properties of the decay kernel $G(t):=\exp(-tB)$ for the particular nonsymmetric but strictly positive $2\times2$-matrix $B=\usebox{\matB}$ with $b>0$. We will see that $G$ may or may not be positive definite, according to the particular choice of  $b$. Thus, our general results obtained for decay kernels defined as matrix functions of symmetric matrices do not carry over to the nonsymmetric case.


\begin{example}[\bfseries  Nonsymmetric matrix exponential decay]Let $B=\usebox{\matB}$,
 where $b>0$ and consider the following decay kernel 
$$G(t)=e^{-tB}=\begin{pmatrix}\exp(-tb)&-t \exp(-tb)\\0&\exp(-tb)\end{pmatrix}.$$
Applying Lemmas~\ref{lemma-eigenvalue2} and~\ref{lem-smooth-noninc-convex}, we easily see that $G$ is not symmetric, not nonnegative, not nonincreasing, and not convex. But $G$ is positive definite   if and only if $b\ge1/2$. To see this, we observe by calculating the inverse Fourier transform that $\widetilde{G}(t)=\int_{\mathbb R} e^{itz}M(z)\,dz $ with 
 $$M(z)=\frac 1\pi  \left(
\begin{array}{cc}
 \frac{b }{b^2+z^2} & \frac{-1}{2( b+iz )^2} \\
 \frac{-1}{2(b-iz)^2} & \frac{b }{b^2+z^2} \\
\end{array}
\right).$$
From Theorem~\ref{bochner-thm-nonsymm} and Lemma~\ref{lemma-eigenvalue2}, $G$ is positive definite if and only if for all $z\in\bR$
$$   \frac{1}{4} \frac{1}{(b^2+z^2)^2} \le\left( \frac{b}{b^2+z^2}\right)^2,   $$
which is in turn equivalent to $1/2\le b$.\hfill$\diamondsuit$
\end{example}


For the following results we no longer require that the decay kernel is given in the  particular form of a matrix function. 


\begin{proposition} \label{prop-exponential2} Let 
$$G(t)=\begin{pmatrix} a_{11} \exp(-b_{11} t) & a_{12} \exp(-b_{12}t) \\ a_{21} \exp(-b_{21}t) & a_{22} \exp(-b_{22}t) \end{pmatrix}
$$
with $a_{11},a_{12},a_{21},a_{22},b_{11},b_{12},b_{21},b_{22}>0$.
\begin{enumerate}
\item $G$ is nonnegative if and only if $\min\{b_{12},b_{21}\} \ge \frac 12 (b_{11}+b_{22})$ and $\frac 14(a_{12}+a_{21})^2 \le a_{11}a_{22}$.
\item $G$ is nonincreasing if and only if $\min\{b_{12},b_{21}\} \ge \frac 12 (b_{11}+b_{22})$ and $\frac 14(a_{12} b_{12} +a_{21} b_{21})^2 \le a_{11}b_{11}a_{22}b_{22}$.
\item $G$ is convex if and only if $\min\{b_{12},b_{21}\} \ge \frac 12 (b_{11}+b_{22})$ and $\frac 14(a_{12} b_{12}^2 +a_{21} b_{21}^2)^2 \le a_{11}b_{11}^2 a_{22}b_{22}^2$.
\item Let $G$ be nonincreasing and $a_{12}=a_{21}$. Then $G$ is positive definite.
\item $G$ is commuting if and only if either $b_{11}=b_{12}=b_{21}=b_{22}$, or $b_{11}=b_{22}$ and $b_{12}=b_{21}$ and $a_{11}=a_{22}$.
\end{enumerate}
\end{proposition}


For the following simpler and symmetric decay kernel, the results follow immediately from the preceding proposition. See Figure~\ref{fig-1} for an illustration of a corresponding optimal strategy.
\begin{figure} 
	\begin{center}
	\includegraphics[width=6cm]{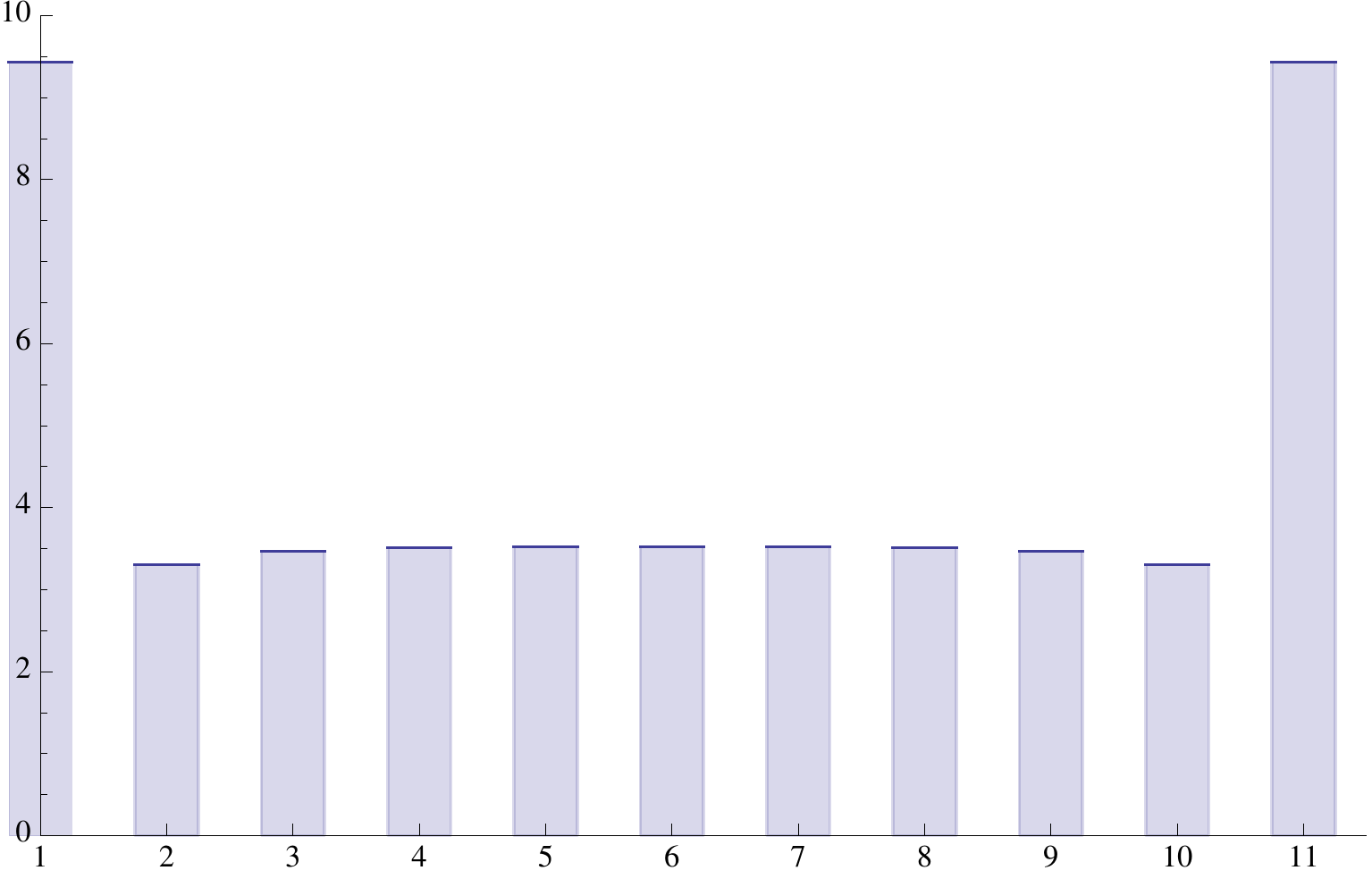}\qquad \qquad \qquad
	\includegraphics[width=6cm]{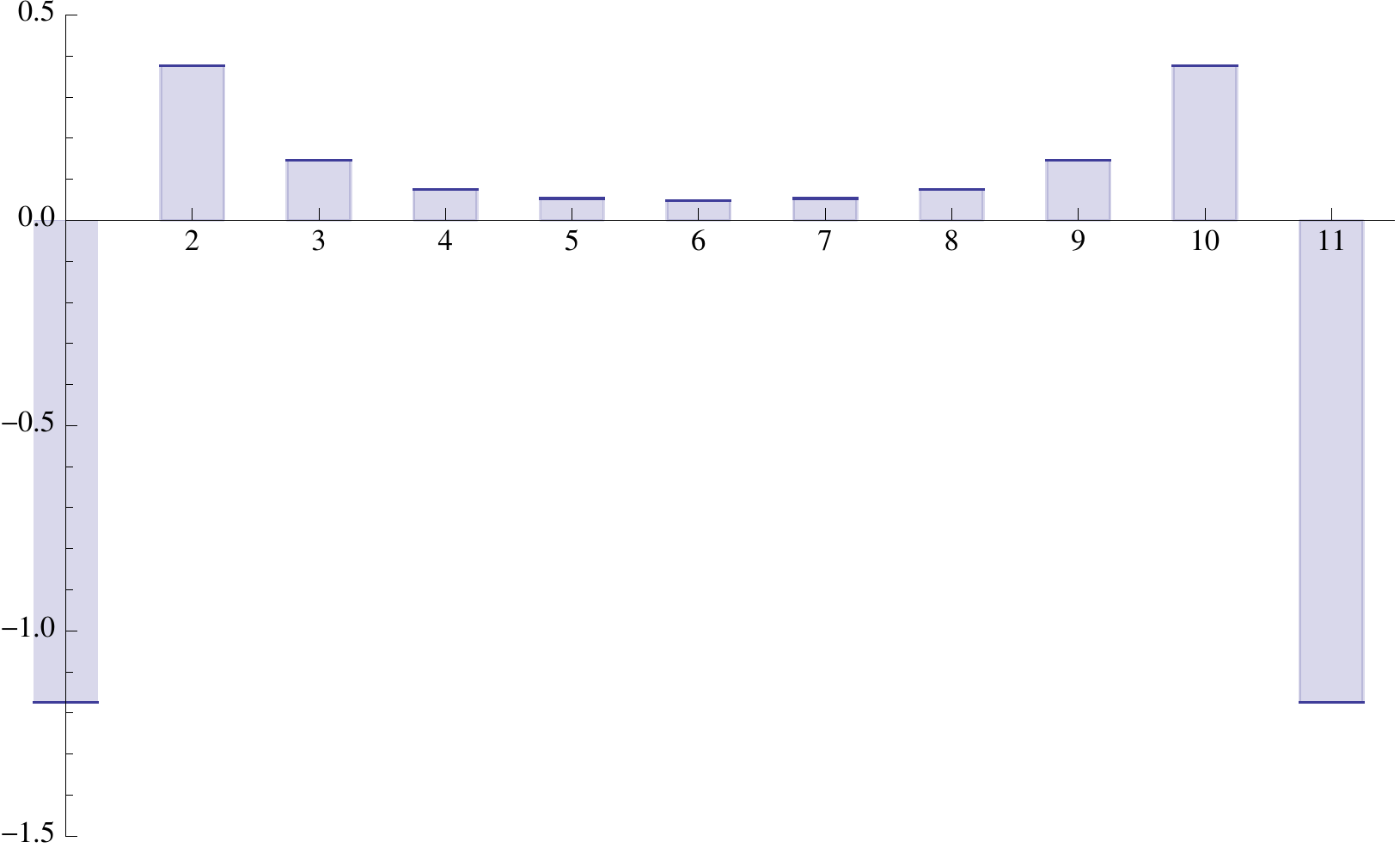}	
	\caption{Optimal strategy $\bm\xi$ for $G$ as in Corollary~\ref{cor-exp-impact} with $\kappa=1$, $\tilde\kappa=1.8$, $\rho=0.3$, $X_0=(-50,1)^\top$, $T=5$, and $ N=11$. Left: $\xi^1_1,\ldots,\xi^1_{11}$, right: $\xi^2_1,\ldots,\xi^2_{11}$.}
	\label{fig-1}
	\end{center}
	\end{figure}


\begin{corollary} \label{cor-exp-impact}
Let $\rho,\kappa,\wt\kappa > 0$ and
$$G(t)=\begin{pmatrix}\exp(-\kappa t) &\rho\exp(-\wt\kappa t)\\\rho\exp(-\wt\kappa t)&\exp(-\kappa t)\end{pmatrix}.$$
\begin{enumerate}
\item $G$ is nonnegative if and only if $\frac{\kappa}{\wt\kappa}\le 1$ and $\rho\le 1$.
\item $G$ is nonincreasing if and only if $\rho\le\frac{\kappa}{\wt\kappa}\le 1$. In this case, it is also nonnegative.
\item $G$ is convex if and only if $\rho\le\frac{\kappa^2}{\wt\kappa^2} \le 1$.
\item If $G$ is nonincreasing, $G$ is positive definite. \item $G$ is commuting.
\end{enumerate}
\end{corollary}


The following proposition shows that we cannot drop the assumption of  symmetry in Theorem~\ref{convex strict pd} in general.


\begin{proposition} \label{prop-convex-nonposdef}
Let 
$$G(t)=\begin{pmatrix} \exp(-t \wedge 1) & \frac 18 \exp(-2(t\wedge 1)) \\ \frac 18 \exp(-3(t\wedge 1)) & \exp(-t\wedge 1) \end{pmatrix}.$$
$G$ is continuous, convex, nonincreasing, and nonnegative, but not positive definite.
\end{proposition}

\subsection{Linear  decay }

In this section, we analyze  linear decay of price impact for $K=2$ assets. 


\begin{proposition}\label{prop-lin-decay}
Let
$$G(t)=\begin{pmatrix} (a_{11} - b_{11} t)^+ & (a_{12} - b_{12} t)^+\\
(a_{21} - b_{21} t)^+ &  (a_{22} - b_{22} t)^+ \end{pmatrix}
$$
with $a_{11},a_{12},a_{21},a_{22},b_{11},b_{12},b_{21},b_{22}>0$.
\begin{enumerate}
\item $G$ is nonnegative if and only if $\max\{\frac{a_{12}}{b_{12}},\frac{a_{21}}{b_{21}}\}\le \min\{\frac{a_{11}}{b_{11}},\frac{a_{22}}{b_{22}}\}$ and $\frac 14(a_{12}+a_{21})^2\le a_{11}a_{22}$.
\item $G$ is nonincreasing if and only if $\max\{\frac{a_{12}}{b_{12}},\frac{a_{21}}{b_{21}}\}\le \min\{\frac{a_{11}}{b_{11}},\frac{a_{22}}{b_{22}}\}$ and $\frac 14 (b_{12} +b_{21})^2 \le b_{11}b_{22}$. 
\item Assume that $\max\{\frac{a_{12}}{b_{12}},\frac{a_{21}}{b_{21}}\}\le \min\{\frac{a_{11}}{b_{11}},\frac{a_{22}}{b_{22}}\}$ and $a_{12}=a_{21}$. Then, $G$ is positive definite if and only if $G$ is symmetric (i.e.\ $a_{12}=a_{21}$ and $b_{12}=b_{21}$), $\frac{a_{11}}{b_{11}}=\frac{a_{12}}{b_{12}}=\frac{a_{22}}{b_{22}}$ and $b_{12}^2 \le b_{11}b_{22}$. In this case, we set $\lambda=\frac{a_{11}}{b_{11}}$ and have
$$G(t)= (\lambda-t)^+ \begin{pmatrix} b_{11} &  b_{12} \\
b_{12} &  b_{22} \end{pmatrix},$$
and $G$ is also nonincreasing, convex, and commuting.
\end{enumerate}
\end{proposition}

\section{Conclusion}\label{conclusion section}

Our goal in this paper was to analyze a linear market impact model with transient price impact for $K$ different risky assets.   We were in particular interested in the question which properties a decay kernel should satisfy so that the corresponding market impact model has certain desirable features and properties.
Let us summarize some of the main messages for the practical application of transient price impact models that can be drawn from our results. 
\begin{enumerate}
\item To exclude price manipulation in the sense of \citet{hubermanstanzl} and to guarantee the existence of optimal strategies, decay kernels should be  positive definite in the sense of Definition \ref{positive definite kernel def} (Propositions \ref{Cge0 pos def prop} and \ref{prop-opt-strat} and Lemma \ref{Pos def funct Hermitian lemma}).
\item Requiring only positive definiteness is typically not sufficient to guarantee that optimal strategies are well-behaved (Figure~\ref{oscillations-fig}). In particular, the nonparametric estimation of decay kernels can be problematic. 
\item Assuming that the decay kernel is symmetric, nonnegative, nonincreasing, convex, and commuting guarantees that optimal strategies have many desirable properties and can easily be  computed  (Propositions \ref{commuting properties Prop} and \ref{prop-orthonormalbasis}). The additional assumption that $t\mapsto x^\top G(t)x$ is nonconstant for all $x\in\bR^K$ guarantees that optimal strategies are unique (Theorem \ref{convex strict pd} and Proposition \ref{prop-opt-strat}).  In this setting, one can also optimize jointly over time grids and strategies and pass to a continuous-time limit.
\item Matrix functions \eqref{matrix function eq}
 provide a  convenient method for constructing decay kernels satisfying the properties from (c). Optimal strategies for matrix exponential decay can be computed in closed form  (Example \ref{example-matrixexp}).
 \item If the   price impact between all asset pairs decay at a uniform rate, then cross-asset impact can be ignored and one can consider each asset individually (Propositions \ref{prop-decomposable-G} and \ref{prop-L-transl}).
\end{enumerate}

\section{Proofs}
\label{sect-proofs}

\begin{proof}[Proof of Lemma~\ref{expected costs lemma}]Using the continuity of $G$ and the right-continuity of $S^0$, we have 
\begin{eqnarray*}-
\mathbb E[\,\mathscr{R}(\bm\xi)\,]&=&\mathbb E\Big[\,\frac12\sum_{k=1}^N\xi_k^\top  (S^{\bm\xi}_{t_k+}+S^{\bm\xi}_{t_k})\,\Big]\\
&=& \mathbb E\bigg[\,\sum_{k=1}^N\xi_k^\top  S^{0}_{t_k}\,\bigg]+\mathbb E\bigg[\,\frac12\sum_{k=1}^N\xi_k^\top G(0)\xi_k+\sum_{k=1}^N\sum_{\ell=1}^{k-1}\xi_k^\top G(t_k-t_\ell)\xi_\ell\,\bigg].
\end{eqnarray*}
From the martingale property of $S^0$ and the requirement that $\sum_{k=1}^N\xi_k=-X_0$ we obtain that
$$ \mathbb E\bigg[\,\sum_{k=1}^N\xi_k^\top  S^{0}_{t_k}\,\bigg]=\mathbb E\Big[\,\sum_{k=1}^N\xi_k^\top  S_T^0\,\Big]=-X_0^\top S^0_0.
$$
Furthermore,
\begin{eqnarray*}
\lefteqn{\frac12\sum_{k=1}^N\xi_k^\top G(0)\xi_k+\sum_{k=1}^N\sum_{\ell=1}^{k-1}\xi_k^\top G(t_k-t_\ell)\xi_\ell}\\
&=&\frac12\sum_{k=1}^N\xi_k^\top \wt G(0)\xi_k+\frac12 \sum_{k=1}^N\sum_{\ell=1}^{k-1}\xi_k^\top \wt G(t_k-t_\ell)\xi_\ell+\frac12 \sum_{k=1}^N\sum_{\ell=1}^{k-1}\xi_\ell^\top  \wt G(t_\ell-t_k)\xi_k=C_\bT(\bm\xi).
\end{eqnarray*}
This proves \eqref{expected costs formula}. 
\end{proof}


\begin{proof}[Proof of Lemma~\ref{deterministic minimizers lemma}]Suppose that a minimizer $\bm\eta\in \mathscr X(\mathbb T,X_0)$ of $\mathbb E[\,C_\bT(\bm\eta)\,]$ exists but that, by way of contradiction, there is no deterministic minimizer of $C_\bT(\cdot)$. Then there can be no $\bm\xi\in \mathscr X_{\text{\rm det}}(\mathbb T,X_0)$ such that $C_\bT(\bm\xi)\le\mathbb E[\,C_\bT(\bm\eta)\,]$. Since $\bm\eta(\omega)\in \mathscr X_{\text{\rm det}}(\mathbb T,X_0)$ for $\mathbb{P}$-a.e. $\omega$, we must thus have $C_\bT(\bm\eta(\omega))> \mathbb E[\,C_\bT(\bm\eta)\,]$ for $\mathbb P$-a.e. $\omega\in\Omega$. But this is a contradiction. The proofs of the remaining assertions are also obvious and left to the reader.
\end{proof}


\begin{proof}[Proof of Proposition~\ref{Cge0 pos def prop}] The equivalence of conditions (a) and (b) follows from Lemma~\ref{deterministic minimizers lemma}.  

To prove the equivalence of (b) and (c), it is sufficient to observe that $C_\bT(\bm\xi)$ is a quadratic form on $\bR^{|\bT|}\otimes\bR^K$, and it is well known that a quadratic form is convex if and only if it is nonnegative.

We next prove the equivalence of (b) and  (d). Clearly, (d) immediately implies (b) using the representation \eqref{C wt G representation} of $C_\bT(\cdot)$ and comparing it with  \eqref{wt G positive def eq} with $z_i\in\bR^K$.  
For the proof of the converse implication, we fix $t_1,\dots, t_N\in\bR$. , Clearly we can assume without loss of generality that $\bT=\{t_1,\dots,t_N\}$ is a time grid in the sense that $0=t_1<t_2<\cdots<t_N$.  An $N$-tuple $\bm\zeta:=(\zeta_1,\dots,\zeta_N)$ with $\zeta_i\in\bC^K$ can be regarded as an element in the tensor product $\bC^N\otimes\bC^K$. Let us thus define the linear map ${\bm L}:\bC^N\otimes\bC^K\to\bC^N\otimes\bC^K$ by 
\begin{equation}\label{Gamma def eq}
{\bm L} \bm\zeta=\Big(\sum_{j=1}^N\wt G(t_1-t_j)\zeta_j,\sum_{j=1}^N\wt G(t_2-t_j)\zeta_j,\dots, \sum_{j=1}^N\wt G(t_N-t_j)\zeta_j\Big).
\end{equation}
  We claim that ${\bm L}$ is Hermitian.  Indeed, for $\bm\eta, \bm\zeta\in\bC^N\otimes\bC^K$, the inner product in $\bC^N\otimes\bC^K$ between $\bm\eta$ and ${\bm L}\bm\zeta$ is given by
\begin{eqnarray*}
\langle \bm\eta,{\bm L}\bm\zeta\rangle=\sum_{i,j=1}^N\eta_i^*\wt G(t_i-t_j)\zeta_j=\sum_{i,j=1}^N\zeta_j^*\wt G(t_i-t_j)^*\eta_i=\sum_{i,j=1}^N\zeta_j^*\wt G(t_j-t_i)\eta_i=\langle \bm\zeta,{\bm L}\bm\eta\rangle,
\end{eqnarray*}
where we have used the fact that $\wt G(t_i-t_j)^*=\wt G(t_i-t_j)^\top=\wt G(t_j-t_i)$. It follows that the restriction of  ${\bm L}$ to $\bR^N\otimes\bR^K $ is symmetric and, due to condition (b),  satisfies $0\le C_\bT(\bm\xi)=\langle\bm\xi,{\bm L}\bm\xi\rangle$ for all $\bm\xi\in\bR^N\otimes\bR^K$.  By the symmetry of $\bm L$ and since $\bm L$ has only real entries, it follows that $\langle \bm\zeta,{\bm L}\bm\zeta\rangle\ge0$ for all $\bm\zeta\in\bC^N\otimes\bC^K$, which is the same as \eqref{wt G positive def eq} and hence yields (d). The remaining assertions are obvious. 
\end{proof}


\begin{proof}[Proof of Proposition~\ref{prop-opt-strat}] We first show the existence of optimal strategies when $\wt G$ is positive definite. We will use the notation introduced in the proof of Proposition~\ref{Cge0 pos def prop}. For $X_0\in\bR^K$ and $\bT$ with $N=|\bT|$ fixed, the minimization of $C_\bT(\bm\xi)$ over $\bm\xi\in\mathscr{X}_{\text{det}}(\bT,X_0)$ is equivalent to the minimization of the symmetric and positive semidefinite quadratic form $\bR^{N}\otimes\bR^K \ni \bm\xi\mapsto\langle \bm\xi, {\bm L}\bm\xi\rangle$ under the equality constraint $A\bm\xi=X_0$, where $\bm L$ is as in \eqref{Gamma def eq} and $A:\bR^N\otimes\bR^K\to\bR^K$ is the linear map   $ A\bm\xi:=\sum_{k=1}^N\xi_k$. For fixed $\bm\eta\in \mathscr{X}_{\text{det}}(\bT,X_0)$, every other $\bm\xi \in\mathscr{X}_{\text{det}}(\bT,X_0)$ can be written as $\bm\xi=\bm\eta+\bm\xi^0$ for some $\bm\xi^0\in \mathscr{X}_{\text{det}}(\bT,0)$. Then, due to the symmetry of $\bm L$,
$$\langle \bm\xi, {\bm L}\bm\xi\rangle = \langle \bm\eta, {\bm L}\bm\eta\rangle+2\langle {\bm L}\bm\eta, \bm\xi^0\rangle+\langle \bm\xi^0, {\bm L}\bm\xi^0\rangle,
$$
and our problem is now equivalent to the unconstraint minimization of the right-hand expression over $\bm\xi^0\in \mathscr{X}_{\text{det}}(\bT,0)$. Clearly, $\bm L\bm\xi^0=0$ implies that also $2\langle {\bm L}\bm\eta, \bm\xi^0\rangle=2\langle \bm\eta, {\bm L}\bm\xi^0\rangle=0$. Therefore the existence of minimizers follows from Section 2.4.2 in  \citet{Boot}.

 The uniqueness of optimal strategies for strictly positive definite $\wt G$ follows immediately from the strict convexity of $\bm\xi\mapsto C_\bT(\bm\xi)$ (see Proposition~\ref{Cge0 pos def prop}). The characterization of optimal strategies through Lagrange multipliers as in \eqref{eq-lagrange-mult} is standard.\end{proof}


\begin{proof}[Proof of Lemma~\ref{Pos def funct Hermitian lemma}] (a) That $H(0)$ is  nonnegative definite follows by taking $N=1$ in \eqref{wt G positive def eq}. To show  $H(-t)=H(t)^*$ for any given $t\in\bR$ we take $N=2$ in \eqref{wt G positive def eq} and let $t_1=0$ and $t_2=t$. It follows from the preceding assertion that  $ z_1^*H(-t) z_2+ z_2^*H(t) z_1$ must be a real number for all $ z_1, z_2\in\bC^K$. Taking $ z_1=c_1e_i$ and $ z_2=c_2e_j$ with $c_k\in\bC$ and $e_\ell$ denoting  the $\ell^{\text{th}}$ unit vector in $\bR^K$ yields that $\overline{c_1}c_2H_{ij}(-t)+c_1\overline{c_2}H_{ji}(t)\in\bR$, where $\overline{c}$ denotes the complex conjugate of $c\in\bC$. Choosing $c_1=c_2=1$ gives $\imag (H_{ij}(-t))=-\imag (H_{ji}(t))$ and $c_1=1, c_2=i$ yields $\real (H_{ij}(-t))=\real (H_{ji}(t))$. 

(b) For $ t_1,\dots,t_N\in\bR$, we define  $\tilde{t}_i=t_N-t_{N+1-i}$ and get from part (a) that for $\bm\zeta\in\bC^N\otimes\bC^K$
\begin{eqnarray*}
0 &\le& \sum_{i,j=1}^N\zeta_{N+1-i}^*  H (\tilde{t}_i-\tilde{t}_j)\zeta_{N+1-j} = \sum_{i,j=1}^N\zeta_{N+1-i}^*  H (-(t_{N+1-i}-t_{N+1-j}))\zeta_{N+1-j}\\
&=& \sum_{i,j=1}^N\zeta_i^*  H(t_i-t_j)^*\zeta_j.
\end{eqnarray*}
\end{proof}


\begin{proof}[Proof of Corollary~\ref{Symmetric Bochner Cor}] For the proof of implication (c)$\Rightarrow$(b), we note first that the matrix $M(d\gamma)$ is symmetric, as $M$ is nonnegative definite and $\bR^{K\times K}$-valued. This implies that the matrix $H(t)$ is also symmetric for all $t$. Next,
the symmetry of $M$ on $\bR$ implies that the imaginary part of $\int_\bR e^{i\gamma t}\,M_{k\ell}(d\gamma)$ is equal to 
$\int_0^\infty\big(\sin(t\gamma)+\sin(-t\gamma)\big)\,M_{k\ell}(d\gamma)=0$. Therefore, $H$ takes values in $\bR^{K\times K}$ and, in turn, in $\bS(K)$. We next
 define a finite $\bR_+$-valued measure $\mu$ through $\mu(A):=x^\top M(A) x$ for $A\in\mathscr{B}(\bR)$. Then the function $t\mapsto x^\top H(t)x$ is the Fourier transform of $\mu$ and hence a positive definite function by Bochner's theorem. 
 
To prove (b)$\Rightarrow$(a), we will establish condition (b) of Theorem~\ref{bochner-thm-nonsymm}. To this end, write $z\in\bC^K$ as $z=x+iw$, where $x,w\in\bR^K$ and $i=\sqrt{-1}$. Then $z^*H(t)z=x^\top H(t)x+w^\top H(t)w$ due to the symmetry of $H(t)$. Hence $t\mapsto z^*H(t)z$ is the sum of two real-valued positive definite functions and therefore positive definite. 

To prove (a)$\Rightarrow$(c), note that each component $H_{k\ell}$ of $H$ is equal to the Fourier transform of the complex measure $M_{k\ell}$. Since $M_{k\ell}$ is uniquely determined through $H_{k\ell}$ and since $H_{k\ell}=H_{\ell k}$ we must have that $M_{k\ell}=M_{\ell k}$. But a symmetric matrix can be nonnegative definite, and hence Hermitian, only if it is real. Therefore we must have $M(A)\in\bS_+(K)$ for all $A\in\mathscr{B}(\bR)$.  Finally, the fact that the symmetric positive definite matrix-valued function $H_{k\ell}$ takes only real values  implies via Lemma \ref{Pos def funct Hermitian lemma} (a) that $H(-t)=H(t)$. Therefore, $H$ is equal to the Fourier transform of  the measure $N(A):=\frac12(M(A)+M(-A))$, $A\in\mathscr{B}(\bR)$. But, since
$M$ is uniquely determined by $H$ according to  Theorem \ref{bochner-thm-nonsymm}, we get that $N=M$, and so $M$ must be symmetric on $\bR$.
\end{proof}


\begin{proof}[Proof of Lemma~\ref{prop-nonneg}]
We assume by way of contradiction that there exist  $x\in\mathbb R^K$, $t^* > 0$ and $\varepsilon>0$ such that $g^x(t):=x^\top G(t)x$ satisfies $g^x(t^*)=-\varepsilon$. We are going to show that the function $g^x$ is not positive definite. Set $t_k = k\cdot t^*$ and $x_k = 1$ for $k \in\mathbb N$. Since $|t_k-t_l| \ge t^*$ for $k\ne l$ and $g^x$ is nonincreasing, we have $g^x(|t_k-t_l|)\le -\varepsilon$ for $k \ne l$. Thus, $\sum_{k,l=1}^n x_k x_l g^x(|t_k-t_l|) \le n g^x(0) - (n^2 - n) \,\varepsilon$. If $n$ is large enough, the latter expression is negative. Thus, $g^x$ is not positive definite, and so $G$ can not be positive definite.
\end{proof}


We now start preparing the proof of Theorem~\ref{convex strict pd} and give a representation of a convex, nonincreasing, nonnegative, and symmetric  function $G:[0,\infty)\to\bR^{K\times K}$. To this end, let us first observe  that, for such $G$,  the limit $G(\infty):=\lim_{t\uparrow \infty}G(t)$ is well defined in the set of nonnegative definite  matrices. Indeed, for any~$x\in\mathbb R^K$, $g^x(t)=x^\top G(t)x$ is a convex, nonincreasing, nonnegative function and thus converges to a limit that we denote by  $g^x(\infty)$.  Let $e_i$ denote the $i^{\text{th}}$ unit vector. By polarization, we have $G_{ij}(t)=\frac{1}{4}(g^{e_i+e_j}(t)-g^{e_i-e_j}(t))$, and this expression converges to $G_{ij}(\infty)=\frac{1}{4}(g^{e_i+e_j}(\infty)-g^{e_i-e_j}(\infty))$. In particular, we have  $g^x(\infty)=x^\top G(\infty)x$ for any $x\in\bR^K$. 


\begin{proposition}\label{convex Fourier Prop}
Let $G:[0,\infty)\to\bR^{K\times K}$ be  convex, nonincreasing, nonnegative, symmetric, and continuous. There exists a nonnegative Radon measure $\mu$ on $(0,\infty)$ and a measurable function $\Lambda:(0,\infty)\to \bS_+(K)$ such that
\begin{equation}\label{convex G integral rep eq}
G(t)=G(\infty) + \int_{(0,\infty)} (r-t)^+ \Lambda(r) \, \mu(dr).
\end{equation}
Furthermore, $G$ is the Fourier transform of the nonnegative definite matrix-valued measure
$$M(d\gamma)=G(\infty)\,\delta_0(d\gamma)+\Phi(\gamma)\,d\gamma,
$$
where $\Phi:\mathbb R\to \bS_+(K)$ is the continuous function given by
$$\Phi(x)=\frac1\pi\int_{(0,\infty)}\frac{1-\cos xy}{x^2}\Lambda(y)\,\mu(dy).
$$
\end{proposition}

\begin{proof}By Lemmas 4.1 and 4.2 in~\citet{gatheralschiedslynko}, we find that for every $x\in\mathbb R^K$ there is a Radon measure $\mu_x$ on $(0,\infty)$ such that
\begin{equation}\label{def_eta_xi}g^x(t)=g^x(\infty)+\int_{(0,\infty)}(r-t)^+ \,\mu_x(dr), \qquad t>0.
\end{equation}
Moreover, 
 $g^x(t)$ is the Fourier transform of the following nonnegative Radon measure on $\mathbb R$
\begin{equation}\label{def_mu_xi}\mu_ x(d t)= x^\top G(\infty)   x \,\delta_0(dt)+\varphi_{ x}(t)\,dt,
\end{equation}
where $\delta_0$ is the Dirac measure concentrated in $0$ and 
\begin{equation}\label{def_phi_xi} \varphi_{ x}(t)=\frac{1}{\pi} \int_{(0,\infty)} \frac{1-\cos ty}{t^2}\,\mu_ x(dy).
\end{equation}
 We consider the finite set  $Z:=\{e_i \pm e_j\,|\,i,j=1,\dots,K\}$ and define $\mu=\sum_{z\in Z}\mu_z$. Then each $\mu_z$ with $z\in Z$ is absolutely continuous with respect to $\mu$ and has the Radon-Nikodym derivative $\lambda_z=d\mu_z/d\mu$. We set 
\begin{equation}\label{wt Gamma}
\Lambda_{ij}(r):= \frac14\Big(\lambda_{e_i+e_j}(r)-\lambda_{e_i-e_j}(r)\Big),\qquad r> 0.
\end{equation}
Clearly, $\Lambda_{ij}(r)=\Lambda_{ji}(r)$, and it remains to prove that $\Lambda$ is $\mu$-a.s.~nonnegative definite. Let $x \in \mathbb R^K$. Since $g^x(t)=\frac14\sum_{i,j=1}^K x_i x_j ( g^{e_i+e_j}(t)-g^{e_i-e_j}(t))$, we necessarily have from~\eqref{def_eta_xi}:
$$g^x(t)-g^x(\infty)=\int_{(0,\infty)}(r-t)^+\, \mu_ x(dr) = \int_{(0,\infty)}(r-t)^+\sum_{i,j=1}^K  x_i  x_j \Lambda_{ij}(r) \, \mu(dr),\qquad t>0. $$
Writing $(r-t)^+$ as $\int_0^\infty\indf{\{t\le s<r\}}\,ds$, integrating by parts, and taking derivatives with respect to $t$ gives 
$$0\le \mu_ x\big((t,+\infty)\big)=\int_{(t,+\infty)} \sum_{i,j=1}^K  x_i  x_j \Lambda_{ij}(r) \, \mu(dr)$$ for any $t\ge 0$ and so $\sum_{i,j=1}^K  x_i  x_j \Lambda_{ij}(r)\ge 0$ for $r \not \in N_ x$, where $N_x$ is such that $\mu(N_ x)=0$. We define $N=\bigcup_{ x \in\mathbb Q^K}N_ x$. Then $\mu(N)=0$ and, by continuity, $\sum_{i,j=1}^K  x_i  x_j \Lambda_{ij}(r)\ge 0$ for all $ x \in\mathbb R^K$ and $r\not \in N$.

Now that we have shown $\mu_ x(dr)=\sum_{i,j=1}^K  x_i  x_j \Lambda_{ij}(r) \, \mu(dr)$, we get \eqref{convex G integral rep eq} from~\eqref{def_eta_xi}.  Next, we obtain from~\eqref{def_phi_xi} that
$$\varphi_{ x}(t)=\frac{1}{\pi} \int_{(0,\infty)} \frac{1-\cos ty}{t^2} \sum_{i,j=1}^K  x_i  x_j \Lambda_{ij}(y)  \, \mu(dy).$$
Again, we define by polarization $\Phi(t)_{ij}=\frac14( \varphi_{e_i+e_j}(t)-\varphi_{e_i-e_j}(t))$. We then have
$$ \Phi(t)=\frac{1}{\pi} \int_{(0,\infty)} \frac{1-\cos ty}{t^2} \Lambda(y)  \, \mu(dy),$$
and  $\varphi_{ x}(t)= x^\top\Phi(t) x$. Together with~\eqref{def_mu_xi}, this gives the claim. 
\end{proof}


\begin{proof}[Proof of Theorem~\ref{convex strict pd}]  From Theorem~\ref{bochner-thm-nonsymm} and the fact that $g^x$ is positive definite for each $x$, we already know that $G$ is positive definite. 
Note also that $G$ cannot be strictly positive definite if there exists some nonzero $x\in\bR^K$ such that $t\mapsto x^\top G(t)x$ is constant, for then the choice $z_1=x$ and $z_2=-x$ gives 
$\sum_{i,j=1}^2z_i^*G(t_i-t_j)z_j=0$ for all $t_1,t_2\in\bR$. 

It thus remains to show that $G$ strictly definite positive if $\zeta^\top G(t)\zeta$ is nonconstant for any $\zeta \in \mathbb R^K$.
We argue first  that, in proving this assertion,  we can assume without loss of generality that $G$ is continuous. To this end, consider again the  functions $g^x(t):=  x^\top G(t)x$ for $x\in\bR^K$. As these functions are convex and nonincreasing, they are continuous on $(0,\infty)$ and admit right-hand limits, $g^x(0+):=\lim_{t\downarrow0}g^x(t)\le g^x(0)$. Using polarization as in the paragraph preceding Proposition \ref{convex Fourier Prop}, we thus conclude that $G$ is continuous on $(0,\infty)$, admits a right-hand limit $G(0+)$, and that $\Delta G(0):=G(0)-G(0+)$ is nonnegative definite. On the other hand, the continuous matrix-valued function $G^{cont}(t):=G(t+)$ also satisfies our assumptions and so will be strictly positive definite when the assertion has been established for continuous matrix-valued  functions. But then $G(t)=G^{cont}(t)+\indf{\{0\}}(t)\Delta G(0)$ will also be strictly positive definite, because $\Delta G(0)$ is nonnegative definite.

 Now, let  $M$, $\Phi$, $\Lambda$, and $\mu$ be as in Proposition~\ref{convex Fourier Prop}.
It follows from this proposition that for $\bm\zeta=( \zeta_1, \zeta_2, \ldots, \zeta_N)\in \mathbb C^N\otimes\mathbb C^K$ and   $ t_1 , t_2 ,\ldots , t_N\in\bR$
\begin{align*}
   \sum_{k,\ell=1}^N \zeta_k^* \widetilde G(t_k - t_\ell) \zeta_\ell  &=    \sum_{k,\ell=1}^N \zeta_k^* \left(\int_{\mathbb R} e^{i(t_k-t_\ell)\gamma} \,M(d{\gamma})\right) \zeta_\ell \nonumber\\
& =  \Big(\sum_{k=1}^N \zeta_k\Big)^* G(\infty)\Big(\sum_{k=1}^N  \zeta_k\Big)+ \int\Big(\sum_{k=1}^N e^{-it_k\gamma}\zeta_k\Big)^*\Phi(\gamma)\Big(\sum_{k=1}^N e^{-it_k\gamma}\zeta_k\Big)\, d\gamma\\
&=v(0)^* G(\infty)v(0)+ \int v(\gamma)^*\Phi(\gamma)v(\gamma)\,d\gamma,\nonumber
\end{align*}
where $v(\gamma):=\sum_{k=1}^N e^{-it_k\gamma}\zeta_k$. We are now going to show that  $\int v(\gamma)^*\Phi(\gamma)v(\gamma)\,d\gamma$  is strictly positive unless $\bm\zeta=0$.  To this end, we note first that the components of the vector field $v(\cdot)$ are holomorphic functions of $\gamma\in\mathbb C$. When $\bm\zeta\neq0$, at least one of these components is nonconstant and hence vanishes only for at most countably many $\gamma\in\bR$.  It follows that $v(\gamma)\neq0$ for all but countably many $\gamma\in\mathbb R$. Moreover, we are going to argue next that the matrix $\Phi(\gamma)$ is strictly positive definite for all but countably many $\gamma\in\mathbb R$.  Thus, $v(\gamma)^*\Phi(\gamma)v(\gamma)>0$ for Lebesgue-almost every $\gamma\in\mathbb R$, and it will follow that $\sum \zeta_k^* \widetilde G(t_k - t_\ell) \zeta_\ell>0$. 

So let us show now that the matrix $\Phi(\gamma)$ is strictly positive definite for all but countably many $\gamma\in\mathbb R$. To this end, we first note that for $z\in\bC^K$
$$g^z (t)=z ^* G(\infty)z +\int_{(0,\infty)}(r-t)^+z ^* \Lambda(r)z \,\mu(dr).
$$
Since the matrix $\Lambda(r)$ is nonnegative definite for all $r$, the fact that $g^z $ is nonconstant for $z \neq0$ implies that 
\begin{equation}\label{zeta Lambda zeta >0}
\int_{(0,\infty)}z ^* \Lambda(r)z \,\mu(dr)>0\qquad\text{for $z \neq0$.}
\end{equation}
Now let $D$ be the set of all $y>0$ such that $\mu(\{y\})>0$, and let
$$\mu_d(E):=\mu(D\cap E)\quad\text{and}\quad \mu_c(E):=\mu(D^c\cap E)
$$
be the discrete and continuous parts of $\mu$, respectively. 
 Clearly,
$$N:=\big\{x\in\mathbb R\,\big|\,\cos xy=1\text{ for some $y\in D$}\big\}=\bigcup_{y\in D}\big\{x\in\mathbb R\,\big|\,\cos xy=1\big\}
$$ 
is at most countable.  Moreover, the set $\{y>0\,|\,\cos xy=1\}$ is a $\mu_c$-nullset for all $x\neq0$. It thus follows that the measure $\frac{1-\cos xy}{x^2}\,\mu(dy)$ is equivalent to $\mu$ for all $x\notin N\cup\{0\}$. Therefore \eqref{zeta Lambda zeta >0} implies that 
$$z ^*\Phi(x)z =\frac1\pi\int_{(0,\infty)}\frac{1-\cos xy}{x^2}z ^* \Lambda(r)z \,\mu(dy)>0
$$
for all $z \neq0$ as long as $x\notin N\cup\{0\}$. This concludes the proof. 
\end{proof}


\begin{proof}[Proof of Proposition~\ref{commuting properties Prop}]
We first prove \eqref{simultaneous diagonalization eq}. To this end, give a constructive proof for the existence of~$O$. For~$t \in \mathbb R$, we write $\mathbb R^K=E_{\lambda_1(t)}^t\oplus \dots\oplus E_{\lambda_{\ell(t)} (t)}^{t} $ for the orthogonal direct sum of the eigenspaces of~$G(t)$ corresponding to the distinct eigenvalues $\lambda_1(t),\dots\lambda_{\ell(t)}(t)$ of $G(t)$. It follows from the commuting property  that the eigenspaces of $G(t)$ are stable under the map~$G(s)$, because $\lambda_i(t)G(s)v=G(t)G(s)v$ if $v \in  E_{\lambda_i(t)}$. 

Let $t_1\ge 0$ and define $D_1=\sum_{i=1}^{\ell(t_1)} (\dim(E_{\lambda_i(t_1)}^{t_1})-1)^+ \le K-1$. If for any $s\ge 0$ and $1 \le i\le {\ell(t_1)} $ there is $\mu_i(s)$ such that $G(s)v=\mu_i(s)v$ for any $v\in E_{\lambda_i(t_1)}^{t_1}$, we are done by considering an orthonormal basis $(v^1_i,\dots,v^{d_i}_i)$ of each eigenspaces $E_{\lambda_i(t_1)}^{t_1}$ and by setting $O=(v^1_1,\dots,v^{d_1}_1,\dots,v^1_\ell,\dots,v^{d_\ell}_\ell)$ for $\ell:=\ell(t_1)$. This is necessarily the case if $D_1=0$. Otherwise, there is $t_2\ge 0$ such that, for at least one $i\in\{1,\dots,\ell(t_1)\}$, the decomposition
$$ E_{\lambda_i(t_1)}^{t_1}=(E_{\lambda_1(t_2)}^{t_2}\cap E_{\lambda_i(t_1)}^{t_1} )\oplus \dots\oplus (E_{\lambda_{\ell(t_2)} (t_2)}^{t_2}\cap E_{\lambda_i(t_1)}^{t_1})$$
is such that
$$\dim (E_{\lambda_i(t_1)}^{t_1})-1 > \sum_{k=0}^{\lambda_{\ell(t_2)}} \big(\dim(E_{\lambda_k(t_2)}^{t_2} \cap E_{\lambda_i(t_1)}^{t_1} )-1\big)^+.$$ We write
$$ \mathbb R^K = \bigoplus_{1\le i_1\le \ell(t_1),1 \le i_2\le {\ell(t_2)}}E_{\lambda_{i_1} (t_1)}^{t_1}\cap E_{\lambda_{i_2} (t_2)}^{t_2}$$  
and have $D_2=\sum_{i_1=1}^{\ell(t_1)} \sum_{i_2=1}^{\ell(t_2)} (\dim(E_{\lambda_{i_1} (t_1)}^{t_1}\cap E_{\lambda_{i_2} (t_2)}^{t_2})-1 )^+ < D_1$. Once again, we are done if there is for any $s\ge 0$, $1 \le i_1\le {\ell(t_1)} $ and $1 \le i_2\le {\ell(t_2)} $, $\mu_{i_1,i_2}(s)$ such that $G(s)v=\mu_{i_1,i_2}(s)v$ for any $v\in E_{\lambda_{i_1} (t_1)}^{t_1}\cap E_{\lambda_{i_2} (t_2)}^{t_2}$. This is the case when $D_2=0$. Otherwise, there is $t_3$ such that 
$$D_3=\sum_{i_1=1}^{\ell(t_1)} \sum_{i_2=1}^{\ell(t_2)}\sum_{i_3=1}^{\ell(t_3)} \big(\dim(E_{\lambda_{i_1} (t_1)}^{t_1}\cap E_{\lambda_{i_2} (t_2)}^{t_2}\cap E_{\lambda_{i_3} (t_3)}^{t_3} )-1 \big)^+ < D_2$$ and we repeat this procedure at most $K$~times to get \eqref{simultaneous diagonalization eq}.

We now prove properties (a)---(e).
Let $v_i$ be the $i^{\text{th}}$ column of $O$. Then $v_i$ is the eigenvector of $G_i(t)$ for the eigenvalue $g_i(t)$.  A given $x\in\mathbb{R}^K$ can be written as $x=\sum_{i=1}^K\alpha_iv_i$. Then $Ox=\sum_{i=1}^K\alpha_ie_i$, where $e_i$ is the $i^{\text{th}}$ unit vector. It follows from \eqref{simultaneous diagonalization eq} that 
$g^x(t)=\sum_{i=1}^K\alpha_i^2g_i(t)$. From here, the assertions (a)---(d) are obvious. Part (e) follows from Proposition~\ref{prop-L-transform}. 
\end{proof}


\begin{proof}[Proof of Proposition~\ref{prop-orthonormalbasis}]
Let $O$ and $g_1,\dots, g_K$ be as in Proposition~\ref{commuting properties Prop}. 
We let $v_1,\dots, v_K$ be the columns of $O$. By Theorem 1 from \citet{alfonsischiedslynko} there is a one-dimensional optimal strategy $\bm \eta^i=(\eta^i_1,\dots,\eta^i_{|\bT|})\in\mathscr{X}_{\text{det}}(\bT,1)$ for the one-dimensional, nonincreasing, nonnegative, and convex decay kernel $g_i$, and $\bm \eta^i$ has only nonnegative components. By part (e) of Proposition~\ref{commuting properties Prop}, $\bm\xi^{(i)}:=\bm \eta^iv_i$ is an optimal strategy for $G$ in $\mathscr{X}_{\text{det}}(\bT,v_i)$ that satisfies condition (a). When $X_0=\sum_{i=1}^K\alpha_iv_i\in\bR^K$ is given,  the strategy with components $\alpha_i{\bm\eta}^iv_i$ is an optimal strategy in $\mathscr{X}_{\text{\rm det}}(\bT,X_0)$ by Proposition~\ref{commuting properties Prop} (e).\end{proof}

\begin{proof}[Proof of Theorem \ref{cont thm}]  The proof of part (a) can be performed along the lines of the proof of Theorem 2.20 from \citet{gatheralschiedslynko} by noting that  Proposition \ref{prop-orthonormalbasis} (b)  yields an upper bound on the number of shares traded by an optimal  strategy $\bm\xi\in\mathscr{X}(\bS,X_0)$ uniformly over  finite time grids  $\bS\subset\bT$:
$$\sum_{1 \le n \le |\bS|, 1\le j\le K} |\xi^{j}_n| \le \sum_{i=1}^K |\alpha_i|\sum_{j=1}^K |v^j_i|.$$
The details are left to the reader. 

 As for part (b), the argument from the proof of Theorem 2.20 in  \citet{gatheralschiedslynko} yields in particular, that 
$
\mathbb{E}[\,C_{\bT_n}(\bm\xi^{(n)})\,]$ decreases to $\mathbb{E}[\,C_{\bT}(X^*)\,]$
if $\bT_1\subset\bT_2\subset\cdots$ are finite time grids such that $\bigcup_n\bT_n$ is dense in $\bT$ and $\bm\xi^{(n)}$ is an optimal strategy in $\mathscr{X}(\bT_n,X_0)$. This proves (b). \end{proof}


\begin{proof}[Proof of Proposition~\ref{prop-decomposable-G}]
 Let $A\in\bR^{K\times K}$ be a symmetric square root of the nonnegative definite matrix $L$ so that $L=A^2=A^\top A$.
  For $t_1,\dots, t_N\in\bR$ and $\zeta_1,\dots,\zeta_N\in\bC^K$ let $\eta_k:=A\zeta_k$. It follows that 
 \begin{eqnarray*}
  \sum_{k,\ell=1}^N \zeta_k^* \wt G(t_k-t_\ell)\zeta_\ell
 = \sum_{k,\ell=1}^N \zeta_k^* A^\top g(|t_k-t_\ell|)A  \zeta_\ell=\sum_{k,\ell=1}^N \eta_k^* g(|t_k-t_\ell|)\eta_\ell=\sum_{i=1}^K\sum_{k,\ell=1}^N \overline\eta_k^i \eta_\ell^i g(|t_k-t_\ell|),
 \end{eqnarray*}
 which is nonnegative since the function $g$ is positive definite. Now let $g$ and $L$ be even strictly positive definite. Then the matrix $A$ is nonsingular and so we have $\eta_1=\cdots=\eta_N=0$  if and only if $\zeta_1=\cdots=\zeta_N=0$. It follows that in all other cases the right-hand side above is strictly positive. 
\end{proof}


\begin{proof}[Proof of Prop.~\ref{prop-L-transl}] Let $\bm\xi\in \mathscr{X}_{\text{det}}(\bT,X_0)$ be an optimal strategy for the decay kernel $G$. By Proposition~\ref{prop-opt-strat} there exists a Lagrange multiplier $\lambda\in\bR^K$ such that
$$\sum_{\ell=1}^N  \wt G(t_k-t_\ell)\xi_\ell=\lambda\qquad\text{for $k=1,\dots, |\bT|$.}
$$
By multiplying both sides of this equation with $L$ we obtain
$$\sum_{\ell=1}^N  \wt G_L(t_k-t_\ell)\xi_\ell=L \lambda\qquad\text{for $k=1,\dots, |\bT|$, }
$$
which, again by Proposition~\ref{prop-opt-strat}, implies that $\bm\xi$ is also optimal for $G_L$.
\end{proof}


\begin{proof}[Proof of Proposition~\ref{prop-L-transform}] 
Since $L$ is invertible, the transformation
$$\bm\xi\longmapsto \bm\xi^L:=(L^{-1}\xi_1,\dots,L^{-1}\xi_{|\bT|})
$$
is a one-to-one map from $ \mathscr{X}_{\text{\rm det}}(\bT,LX_0)$ to $ \mathscr{X}_{\text{\rm det}}(\bT,X_0)$. We also have
$$0\le\sum_{k,\ell}\xi_k^\top\wt G(t_k-t_\ell)\xi_\ell=\sum_{k,\ell}(\xi^L_k)^\top\wt G^L(t_k-t_\ell)\xi_\ell^L
$$
for all $\bT$ and $\bm\xi$. Minimizing the two sums over the respective classes of strategies yields the result. 
\end{proof}


To study the examples for $K=2$ assets, we will frequently use the following simple lemma. 


\begin{lemma} \label{lemma-eigenvalue2}
\newsavebox{\Genmat}
\savebox{\Genmat}{$\left(\begin{smallmatrix}a&b\\c&d\end{smallmatrix}\right)$}
 For $a,d \ge 0$ and $b,c\in\bR$, the matrix  $M:=\usebox{\Genmat}$ is nonnegative if and only if $\frac 14(b+c)^2 \le ad$. When $b\in \mathbb C$,  the Hermitian matrix $N:=\left(\begin{smallmatrix}a&b\\\bar{b}&d\end{smallmatrix}\right)$ is  nonnegative definite if and only if $|b|^2\le ad$.
\end{lemma}
\begin{proof}
The matrix $M$ is nonnegative if and only if its symmetrization, $\wt M:=\frac 12(M+M^\top)$, is positive definite. Since a symmetric matrix is nonnegative definite if and only if all its leading principle minors are nonnegative and since $\det \wt M=ad-\frac14(b+c)^2$, the result follows. 
In the Hermitian case,  the same condition on the minors holds. 
\end{proof}


\begin{lemma}\label{lem-smooth-noninc-convex}
Let $G:[0,\infty) \rightarrow \mathbb R^{K\times K}$ and assume that $G(t)=\int_0^t \Lambda(s)ds$ for $t\ge 0$. Then, $G$ is nonincreasing if and only if $-\Lambda(s)$ is  nonnegative for a.e.~$s$. If in addition $\Lambda$ is piecewise continuous, then $G$ is convex if and only if $-\Lambda$ is nonincreasing. 
\end{lemma}
\begin{proof}
The function $G$ is nonincreasing if and only if for any $\zeta \in \mathbb R^K$, 
$\int_0^t \zeta^\top \Lambda(s)\zeta ds$ is nonincreasing. This gives $\zeta^\top \Lambda(s)\zeta\ge 0$ for $s\not \in N_{\zeta}$, where $N_{\zeta}$ is a set with zero Lebesgue measure. We define $N=\cup_{\zeta \in \mathbb Q^K}N_{\zeta}$ and have by continuity  $\zeta^\top \Lambda(s)\zeta\ge 0$ for any $s\not \in N$, $\zeta \in \mathbb R^K$. The converse implication as well as the other equivalence are obvious. 
\end{proof}


\begin{proof}[Proof of Proposition~\ref{prop-exponential2}]
(a): By Lemma~\ref{lemma-eigenvalue2}, $G$ is nonnegative if and only if for every $t \ge 0$
$$\frac 14(a_{12}\exp(-b_{12}t)+a_{21}\exp(-b_{21}t))^2\le a_{11}\exp(-b_{11}t)a_{22}\exp(-b_{22}t).$$
That is, if and only if
$$\frac 14(a_{12}^2\exp(-2b_{12}t)+2a_{12}a_{21}\exp(-(b_{12}+b_{21})t)+a_{21}^2\exp(-2b_{21}t))\le a_{11}a_{22}\exp(-(b_{11}+b_{22})t).$$
If $G$ is nonnegative, taking $t=0$ shows $\frac 14(a_{12}+a_{21})^2 \le a_{11}a_{22}$, while sending $t\rightarrow\infty$ shows $\min\{b_{12},b_{21}\} \ge \frac 12 (b_{11}+b_{22})$. Conversely, if these inequalities hold, $G$ is nonnegative.

(b): $G$ is continuously differentiable. By Lemma~\ref{lem-smooth-noninc-convex}, $G$ is hence nonincreasing if and only if for every $t\ge 0$
$$-G'(t)=\begin{pmatrix} a_{11}b_{11} \exp(-b_{11} t) & a_{12}b_{12} \exp(-b_{12}t) \\ a_{21}b_{21} \exp(-b_{21}t) & a_{22}b_{22} \exp(-b_{22}t) \end{pmatrix}$$
is nonnegative. Analogously to (a), the result follows.

(c): Analogously to (b), by  Lemma~\ref{lem-smooth-noninc-convex} $G$ is convex if and only if for every $\ge 0$ its second derivative
$$G''(t)=\begin{pmatrix} a_{11}b_{11}^2 \exp(-b_{11} t) & a_{12}b_{12}^2 \exp(-b_{12}t) \\ a_{21}b_{21}^2 \exp(-b_{21}t) & a_{22}b_{22}^2 \exp(-b_{22}t) \end{pmatrix}$$
is nonnegative. The result follows analogously to (a).

(d): The assumption $a_{12}=a_{21}$ gives the continuity of $\widetilde G$. We have that $\widetilde G(t)=\int_{\mathbb R} e^{i \gamma t} M(d\gamma)$, where $M(d\gamma)=\frac 1{2\pi} \Lambda(\gamma)\,d\gamma$ with the Hermitian matrix
$$\Lambda(\gamma)=\begin{pmatrix} 2\frac{a_{11}b_{11}}{b_{11}^2+\gamma^2} & \frac{a_{12}}{b_{21}-i\gamma}+\frac{a_{12}}{b_{12}+i\gamma}\\ \frac{a_{12}}{b_{12}-i\gamma}+\frac{a_{12}}{b_{21}+i\gamma} & 2\frac{a_{22}b_{22}}{b_{22}^2+\gamma^2} \end{pmatrix}.
$$
From Theorem~\ref{bochner-thm-nonsymm}, $G$ is positive definite if and only if the matrix $\Lambda(\gamma)$ is nonnegative for almost all $\gamma\in\mathbb R$. According to Lemma~\ref{lemma-eigenvalue2}, this is equivalent to
$$\frac{a_{12}^2 (b_{12}+b_{21})^2}{(b_{12}^2+\gamma^2)(b_{21}^2+\gamma^2)}\le 4\frac{a_{11}b_{11}}{b_{11}^2+\gamma^2} \frac{a_{22}b_{22}}{b_{22}^2+\gamma^2}.$$
This condition is in turn equivalent to
$${a_{12}^2 (b_{12}+b_{21})^2}(b_{11}^2+\gamma^2)(b_{22}^2+\gamma^2)\le  4{a_{11}b_{11}}{a_{22}b_{22}}{(b_{12}^2+\gamma^2)(b_{21}^2+\gamma^2)}.$$
Comparing the coefficients for $\gamma^0$, $\gamma^2$ and $\gamma^4$, we see that it is sufficient to have
\begin{eqnarray}
a_{12}^2 (b_{12}+b_{21})^2 b_{11} b_{22}&\le& 4 a_{11}a_{22}b_{12}^2b_{21}^2 \label{pf-exponential2-1}\\
a_{12}^2 (b_{12}+b_{21})^2 (b_{11}^2+b_{22}^2)&\le&4a_{11}b_{11}a_{22}b_{22}(b_{12}^2+b_{21}^2) \label{pf-exponential2-2}\\
a_{12}^2 (b_{12}+b_{21})^2&\le&4a_{11}b_{11}a_{22}b_{22} \label{pf-exponential2-3}.
\end{eqnarray}
Note that (\ref{pf-exponential2-3}) follows immediately from (b), since $G$ is nonincreasing and $a_{12}=a_{21}$. To show (\ref{pf-exponential2-1}), note that $\sqrt{b_{11}b_{22}}\le\frac 12(b_{11}+b_{22})\le \min\{b_{12},b_{21}\}$, so $b_{11}b_{22}\le(\min\{b_{12},b_{21}\})^2\le b_{12}b_{21}$. Together with (\ref{pf-exponential2-3}) the result follows. Now, we claim that $b_{11}^2+b_{22}^2\le b_{12}^2+b_{21}^2$, which together with  (\ref{pf-exponential2-3}) gives  (\ref{pf-exponential2-2}). To see this, we define $m=\min\{b_{12},b_{21}\}$ and assume without loss of generality that $b_{11}\le b_{22}$. Since $\frac 12(b_{11}+b_{22})\le m$, we have $b_{11}\in (0,m]$ and $b_{11}^2+b_{22}^2\le (2m-b_{11})^2+ b_{11}^2 \le 2m^2$ because the polynomial function $x\in [0,m]\mapsto(2m-x)^2+ x^2$ reaches its maximum for $x\in\{0,m\}$.

(e): We find that the left upper entry of $G(0)G(t)-G(t)G(0)$ is $a_{12}a_{21}(e^{-b_{21}t}-e^{-b_{12}t})$, so $G(0)G(t)=G(t)G(0)$ implies $b_{12}=b_{21}$. Given that, a direct calculation shows that $G(0)G(t)=G(t)G(0)$ is equivalent to $a_{11}(e^{-b_{11}t}-e^{-b_{12}t})+a_{22}(e^{-b_{12}t}-e^{-b_{22}t})=0$. If $a_{11}=a_{22}$, this implies $b_{11}=b_{22}$. If $a_{11}\neq a_{22}$, by the equivalent equation $a_{22}-a_{11}=a_{22}e^{-(b_{22}-b_{12})t}-a_{11}e^{-(b_{11}-b_{12})t}$ we see that $b_{11}=b_{22}=b_{12}$.

Conversely, if either $a_{11}=a_{22}$ and $b_{12}=b_{21}$ and $b_{11}=b_{22}$, or $b_{11}=b_{12}=b_{21}=b_{22}$, a direct calculation shows that $G(s)G(t)=G(t)G(s)$ for all $s,t \ge 0$.
\end{proof}


\begin{proof}[Proof of Proposition~\ref{prop-convex-nonposdef}]
$G$ is obviously continuous and Proposition~\ref{prop-exponential2} yields that $G$ is nonnegative, nonincreasing and convex since $-G'$ is nonincreasing. 

To show that $G$ is not positive definite, using Mathematica we find that $G(t)=\int_{\mathbb R} e^{i\gamma t} M(d\gamma)$, where $M(d\gamma)=C \Lambda(\gamma)\,d\gamma + D \delta_0 (d\gamma)$ with a constant $C>0$, a matrix $D\in\mathbb R^{2\times 2}$, the Dirac measure $\delta_0$ at $0$ and $\Lambda(\gamma)$ given by
$$\begin{pmatrix}
 \frac{2 e^2 (-\cos (\gamma ) \gamma +e \gamma -\sin (\gamma ))}{\gamma ^3+\gamma } &
   \frac{5 e^3 \gamma -((3+2 e) \gamma +6 i (-1+e)) \cos (\gamma )+(i (-3+2 e) \gamma
   -6 (1+e)) \sin (\gamma )}{8 \gamma  (\gamma  (\gamma +i)+6)} \\
 \frac{5 e^3 \gamma -(3 (\gamma +2 i)+2 e (\gamma -3 i)) \cos (\gamma )+(-2 i e \gamma
   +3 i \gamma -6 e-6) \sin (\gamma )}{8 \gamma  (\gamma +2 i) (\gamma -3 i)} & \frac{2
   e^2 (-\cos (\gamma ) \gamma +e \gamma -\sin (\gamma ))}{\gamma ^3+\gamma } \\
\end{pmatrix}.$$

If $G$ was positive definite, then all eigenvalues of $\Lambda(\gamma)$ would be positive for $\gamma \ne 0$. But using Mathematica we find that one eigenvalue of $\Lambda(\gamma)$ is 
\begin{eqnarray*}&& \frac1{8 \left(\gamma
   ^2+4\right) \left(\gamma ^2+9\right) \left(\gamma ^3+\gamma \right)^2}
\Big(16 e^3 \left(\gamma ^2+1\right) \left(\gamma ^2+4\right) \left(\gamma ^2+9\right)
   \gamma ^2\\
&& -16 e^2 \left(\gamma ^2+1\right) \left(\gamma ^2+4\right) \left(\gamma
   ^2+9\right) \gamma  (\sin (\gamma )+\gamma  \cos (\gamma ))\\
&&-\Big(\gamma ^2
   \left(\gamma ^2+1\right)^4 \left(\gamma ^2+4\right) \left(\gamma ^2+9\right)
   \Big(\left(9+4 e^2+25 e^6\right) \gamma ^2\\
& &-10 e^3 (3+2 e) \gamma ^2 \cos (\gamma
   )+12 e \left(\gamma ^2-6\right) \cos (2 \gamma )\\
& &-60 e \gamma  \sin (\gamma )
   \left(-2 \cos (\gamma )+e^3+e^2\right)+36 \left(1+e^2\right)\Big)\Big)^{\frac 12}\Big),
\end{eqnarray*}
which is negative for all $\gamma$ with $0<|\gamma|<0.02$. So $G$ is not positive definite.
\end{proof}


\begin{proof}[Proof of Proposition~\ref{prop-lin-decay}] (a): By Lemma~\ref{lemma-eigenvalue2}, $G$ is nonnegative if and only if for every $t \ge 0$
$$\frac 14((a_{12} - b_{12} t)^++(a_{21} - b_{21} t)^+)^2\le (a_{11} - b_{11} t)^+(a_{22} - b_{22} t)^+.$$
Assume that $G$ is nonnegative. Choosing $t=0$ yields $\frac 14(a_{12}+a_{21})^2\le a_{11}a_{22}$. Choosing $t=\min\{\frac{a_{11}}{b_{11}},\frac{a_{22}}{b_{22}}\}$ yields that the right-hand side of the preceding equation is zero. So the left-hand side has to be zero which implies that $\max\{\frac{a_{12}}{b_{12}},\frac{a_{21}}{b_{21}}\} \le t$.

Conversely, assume that $\frac 14(a_{12}+a_{21})^2\le a_{11}a_{22}$ and $\max\{\frac{a_{12}}{b_{12}},\frac{a_{21}}{b_{21}}\} \le \min\{\frac{a_{11}}{b_{11}},\frac{a_{22}}{b_{22}}\}$. So for any $t \ge 0$, we have that $\max\{(1-\frac{b_{12}}{a_{12}}t)^+,(1-\frac{b_{21}}{a_{21}}t)^+\}\le \min\{(1-\frac{b_{11}}{a_{11}}t)^+,(1-\frac{b_{22}}{a_{22}}t)^+\}$. Thus,
\begin{eqnarray*}
\frac 14((a_{12} - b_{12} t)^++(a_{21} - b_{21} t)^+)^2 &=&\frac 14\left(a_{12}\left(1 - \frac{b_{12}}{a_{12}} t\right)^++a_{21}\left(1 - \frac{b_{21}}{a_{21}} t\right)^+\right)^2\\
&\le& \frac 14\left((a_{12}+a_{21})\max\left\{\left(1-\frac{b_{12}}{a_{12}}t\right)^+,\left(1-\frac{b_{21}}{a_{21}}t\right)^+\right\}\right)^2\\
&\le& a_{11}a_{22}\left(\min\left\{\left(1-\frac{b_{11}}{a_{11}}t\right)^+,\left(1-\frac{b_{22}}{a_{22}}t\right)^+\right\}\right)^2\\
&\le& a_{11}a_{22}\left(1-\frac{b_{11}}{a_{11}}t\right)^+\left(1-\frac{b_{22}}{a_{22}}t\right)^+\\
&=&(a_{11}-b_{11}t)^+(a_{22}-b_{22}t)^+.
\end{eqnarray*}
So $G$ is nonnegative.

(b): $G$ is absolutely continuous with derivative
$$G'(t)=\begin{pmatrix}-b_{11} \indf{\{t<\frac{a_{11}}{b_{11}}\}} &-b_{12} \indf{\{t<\frac{a_{12}}{b_{12}}\}}\\-b_{21} \indf{\{t<\frac{a_{21}}{b_{21}}\}}&-b_{22} \indf{\{t<\frac{a_{22}}{b_{22}}\}}\end{pmatrix}$$
By Lemmas~\ref{lemma-eigenvalue2} and~\ref{lem-smooth-noninc-convex}, $G$ is nonincreasing if and only if for almost all $t>0$
$$
\frac 14(b_{12} \indf{\{t<\frac{a_{12}}{b_{12}}\}}+b_{21} \indf{\{t<\frac{a_{21}}{b_{21}}\}})^2\le b_{11} \indf{\{t<\frac{a_{11}}{b_{11}}\}}b_{22} \indf{\{t<\frac{a_{22}}{b_{22}}\}}.
$$
Assume $G$ is nonincreasing. Then choosing $t$ small enough shows $\frac 14 (b_{12} +b_{21})^2 \le b_{11}b_{22}$. Choosing any $t \ge \min\{\frac{a_{11}}{b_{11}},\frac{a_{22}}{b_{22}}\}$ yields that the right-hand side of the preceding equation is zero. So the left-hand-side has to be zero which implies $\max\{\frac{a_{12}}{b_{12}},\frac{a_{21}}{b_{21}}\} \le \min\{\frac{a_{11}}{b_{11}},\frac{a_{22}}{b_{22}}\}$.

Conversely, if $\frac 14 (b_{12} +b_{21})^2 \le b_{11}b_{22}$ and $\max\{\frac{a_{12}}{b_{12}},\frac{a_{21}}{b_{21}}\} \le \min\{\frac{a_{11}}{b_{11}},\frac{a_{22}}{b_{22}}\}$, it is obvious that $G$ is nonincreasing.

(c): By computing the inverse Fourier transform, we easily get that for $a,b_+,b->0$,
 $$\indf{\{t\ge 0\}} (a-b_+t)^+ + \indf{\{t< 0\}} (a+b_-t)^+=\int_{\mathbb R} e^{i \gamma t} \frac{1}{2\pi \gamma^2} \left( b_+ (1-e^{-\frac{a\gamma }{b_+}})+ b_- (1-e^{\frac{a\gamma }{b_-}}) \right)d \gamma. $$
Thanks to the assumption $a_{12}=a_{21}$, $\widetilde{G}$ is continuous and $\widetilde G(t)=\int_{\mathbb R} e^{i\gamma t} M(d\gamma)$ with $M(d\gamma)=\frac 1{2 \pi \gamma^2}\Lambda(\gamma)d\gamma$, with the Hermitian matrix
$$\Lambda(\gamma)=\begin{pmatrix} 2 b_{11}(1-\cos(\frac{a_{11}}{b_{11}} \gamma)) & b_{12}(1-e^\frac{-i a_{12}\gamma}{b_{12}})+b_{21}(1-e^\frac{i a_{12}\gamma}{b_{21}})  \\ b_{21}(1-e^\frac{-i a_{12}\gamma}{b_{21}})+b_{12}(1-e^\frac{i a_{12}\gamma}{b_{12}}) & 2 b_{22}(1-\cos(\frac{a_{22}}{b_{22}} \gamma)) \end{pmatrix}.$$
From Theorem~\ref{bochner-thm-nonsymm}, $G$ is positive definite if and only if $\Lambda(\gamma)$ is positive definite for every $\gamma\in\mathbb R$. Using Lemma~\ref{lemma-eigenvalue2}, $\Lambda(\gamma)$ is positive definite if and only if 
$$| b_{12}(1-e^\frac{-i a_{12}\gamma}{b_{12}})+b_{21}(1-e^\frac{i a_{12}\gamma}{b_{21}})|^2 \le b_{11}(1-\cos(\frac{a_{11}}{b_{11}} \gamma)) b_{22}(1-\cos(\frac{a_{22}}{b_{22}} \gamma)),$$
i.e.\ if and only if
\begin{eqnarray*}
\left(b_{12}(1-\cos(\frac{a_{12}}{b_{12}}\gamma))+b_{21}(1-\cos(\frac{a_{12}}{b_{21}}\gamma))\right)^2
+\left(b_{12}\sin(\frac{a_{12}}{b_{12}}\gamma)-b_{21}\sin(\frac{a_{12}}{b_{21}}\gamma)\right)^2\\
\le 4 b_{11}(1-\cos(\frac{a_{11}}{b_{11}} \gamma)) b_{22}(1-\cos(\frac{a_{22}}{b_{22}} \gamma)),
\end{eqnarray*}
which is equivalent to
$$ \frac{a_{11}}{b_{11}}=\frac{a_{12}}{b_{12}}=\frac{a_{21}}{b_{21}}=\frac{a_{22}}{b_{22}}, \ b_{12}^2\le b_{11}b_{22}.$$
One implication is obvious. To see the other one, we apply the condition to $\gamma=2\pi \frac{b_{11}}{a_{11}}$, which gives $\frac{a_{12}b_{11}}{b_{12}a_{11}} \in \mathbb N$ and $\frac{a_{21}b_{11}}{b_{21}a_{11}} \in \mathbb N$, and thus $\frac{a_{12}b_{11}}{b_{12}a_{11}}=\frac{a_{21}b_{11}}{b_{21}a_{11}}=1$ since $\max\{\frac{a_{12}}{b_{12}},\frac{a_{21}}{b_{21}}\}\le \min\{\frac{a_{11}}{b_{11}},\frac{a_{22}}{b_{22}}\}$ by assumption.
Similarly, considering $\gamma=2\pi \frac{b_{22}}{a_{22}}$ gives $\frac{a_{12}b_{22}}{b_{12}a_{22}}=\frac{a_{21}b_{22}}{b_{21}a_{22}}=1$. In particular, $b_{12}=b_{21}$ and the condition for $\gamma=0$ gives the inequality on $b$'s. The remainder is obvious.
\end{proof}

\parskip-0.5em\renewcommand{\baselinestretch}{0.9}\small
\bibliography{literature}{}
\bibliographystyle{abbrvnat}

\end{document}